\documentclass[letterpaper]{article} %
\usepackage[draft]{aaai2027}  %
\usepackage[hyphens]{url}  %
\usepackage{graphicx} %
\usepackage{natbib}  %
\usepackage{caption} %
\usepackage{algorithm}
\usepackage{algorithmic}

\usepackage{newfloat}
\usepackage{listings}
\DeclareCaptionStyle{ruled}{labelfont=normalfont,labelsep=colon,strut=off} %
\floatstyle{ruled}
\newfloat{listing}{tb}{lst}{}
\floatname{listing}{Listing}

\usepackage{booktabs}

\newif\ifshort
\shortfalse

\newif\iftwincoversolved
\twincoversolvedfalse

\newif\ifCutLemmaLMRMIpolytree
\CutLemmaLMRMIpolytreefalse

\usepackage{amssymb}
\usepackage{amsfonts}
\usepackage{mathtools}
\usepackage{amsthm}
\usepackage{stmaryrd}
\usepackage{thm-restate}
\usepackage{hyperref}
\hypersetup{
    psdextra,
    colorlinks=true,
    citecolor=green!40!black,
    linkcolor=blue,
    urlcolor=blue!80!black
}
\usepackage{cleveref}
\usepackage{pict2e}
\usepackage{tikz}
\usetikzlibrary{arrows.meta,positioning}
\usepackage{xcolor}
\usepackage{comment}
\usepackage{xspace}
\usepackage{nicefrac}
\usepackage[inline]{enumitem}

\definecolor{blue}{HTML}{004488}
\definecolor{red}{HTML}{EE99AA}

\newcommand{\probName}[1]{\textsc{#1}\xspace}
\def\ListMI{\probName{Max List Majority Illusion}}
\def\qMI{\probName{$q$-Majority Illusion}}
\newcommand{\agents}{N}

\newcommand{\N}{\mathbb{N}}
\newcommand{\Z}{\mathbb{Z}}
\newcommand{\myvec}[1]{\mathbf{#1}}
\newcommand{\myveccomp}[1]{#1}
\newcommand{\mymatrix}[1]{\mathbf{#1}}

\newcommand{\Oh}[1]{{\mathcal{O}\left(#1\right)}}

\NewDocumentCommand{\cc}{ O{} O{} m }{\mbox{%
    \expandafter\ifx\expandafter\relax\detokenize{#2}\relax\else{#2-}\fi%
    \textsf{#3}%
    \expandafter\ifx\expandafter\relax\detokenize{#1}\relax\else{-#1}\fi%
    }\xspace}

\newcommand{\NPc}{\cc[complete]{NP}}

\newcommand{\FPT}{\cc{FPT}}
\newcommand{\XP}{\cc{XP}}

\newcommand{\Wh}[1][1]{\cc[hard]{W[#1]}}
\newcommand{\Whness}[1][1]{\cc[hardness]{W[#1]}}

\newtheorem{theorem}{Theorem}
\newtheorem{lemma}{Lemma}
\newtheorem{definition}{Definition}

\newtheorem{claim}{Claim}
\Crefname{claim}{Claim}{Claims}
\newtheorem{corollary}{Corollary}
\newenvironment{claimproof}[1]{\noindent\emph{Proof.}\hspace{0.15cm}#1}{\hfill$\blacktriangleleft$\medskip}

\newcommand{\proofsubparagraph}[1]{\medskip{{\emph{#1.}}}\hspace{0.25cm}}

\newcommand{\linkproof}[1]{%
    $\star$%
}

\newcommand{\opentriangle}{%
  \raisebox{0.2pt}{\makebox[0.77778em]{%
    \setlength{\unitlength}{0.6em}%
    \linethickness{0.4pt}\roundjoin
    \begin{picture}(1,1)
    \polygon(0,0)(1,0)(1,1)
    \end{picture}%
  }}%
}\newenvironment{proofsketch}{\par
  \pushQED{\hfill \opentriangle}%
  \normalfont \topsep 12pt\relax
  \trivlist
  \item[\hskip\labelsep
        {\itshape Proof sketch.}]\ignorespaces
}{%
  \popQED\endtrivlist
}

\def\shortred{\boldsymbol{r}}
\def\red{\ensuremath{\mathsf{red}}}
\def\shortblue{\boldsymbol{b}}
\def\blue{\ensuremath{\mathsf{blue}}}
\def\RedNo{\#{\shortred}}
\def\BlueNo{\#{\shortblue}}
\def\rblab{f}
\def\mapvc{g}
\newcommand{\ol}[1]{\overline{#1}}
\newcommand{\insubtree}[1]{\agents^{#1}}
\def\optcol{f}
\newcommand\umi[1][\rblab]{\mu_{#1}}
\newcommand\mrn[1][\rblab]{\hat\mu_{#1}}
\newcommand\vset[1]{{#1}}

\def\neighb{\mathcal{N}}
\def\root{\rho}
\DeclareMathOperator{\parent}{parent}
\DeclareMathOperator{\children}{children}

\newcommand{\DP}{\operatorname{DP}}

\def\td{\operatorname{td}}

\def\dcg{\operatorname{dcg}}
\def\fvs{\operatorname{fvs}}
\def\fes{\operatorname{fes}}
\def\dds{\operatorname{dds}}
\def\ddp{\operatorname{ddp}}
\def\parent{\operatorname{parent}}
\def\vi{\operatorname{vi}}
\def\vc{\operatorname{vc}}
\def\maxcc{\operatorname{cc}^{\max}}
\DeclareMathOperator{\cdn}{cdn}

\def\Aodd{N_{\mathsf{odd}}}
\def\Aeven{N_{\mathsf{even}}}
\def\ssize{s^\star}
\def\ssum{\myvec{s}_{\Sigma}}
\def\calT{\mathcal{T}}
\def\calS{\mathcal{S}}
\def\calP{\mathcal{P}}
\def\Podd{\calP_{\mathsf{odd}}}
\def\Peven{\calP_{\mathsf{even}}}
\def\opt{\textup{opt}}

\makeatletter
\providecommand*{\cupdot}{%
  \mathbin{%
    \mathpalette\@cupdot{}%
  }%
}
\newcommand*{\@cupdot}[2]{%
  \ooalign{%
    $\m@th#1\cup$\cr
    \sbox0{$#1\cup$}%
    \dimen@=\ht0 %
    \sbox0{$\m@th#1\cdot$}%
    \advance\dimen@ by -\ht0 %
    \dimen@=.5\dimen@
    \hidewidth\raise\dimen@\box0\hidewidth
  }%
}

\providecommand*{\bigcupdot}{%
  \mathop{%
    \vphantom{\bigcup}%
    \mathpalette\@bigcupdot{}%
  }%
}
\newcommand*{\@bigcupdot}[2]{%
  \ooalign{%
    $\m@th#1\bigcup$\cr
    \sbox0{$#1\bigcup$}%
    \dimen@=\ht0 %
    \advance\dimen@ by -\dp0 %
    \sbox0{\scalebox{2}{$\m@th#1\cdot$}}%
    \advance\dimen@ by -\ht0 %
    \dimen@=.5\dimen@
    \hidewidth\raise\dimen@\box0\hidewidth
  }%
}
\makeatother

\newcommand{\Yes}{\emph{yes}\xspace}
\newcommand{\No}{\emph{no}\xspace}
\def\labeling{labeling}

\newcommand{\TODOnote}{{\bfseries\textcolor{red}{TODO}}}

\makeatletter
\let\oldparagraph=\paragraph
\renewcommand\paragraph[1]{\oldparagraph{#1.}}
\makeatother

\title{What Makes Majority Illusion Easy to Detect?}
\author {
    Šimon Schierreich\textsuperscript{\rm 1},
    Ildikó Schlotter\textsuperscript{\rm 2,3}
}
\affiliations {
    \textsuperscript{\rm 1}Czech Technical University in Prague\\
    \textsuperscript{\rm 2}ELTE Centre for Economic and Regional Studies\\
    \textsuperscript{\rm 3}Budapest University of Technology and Economics\\
    schiesim@fit.cvut.cz, schlotter.ildiko@krtk.elte.hu
}

\begin{document}

\maketitle

\begin{abstract}
    Majority illusion is an undesirable phenomenon in social networks in which agents incorrectly perceive a minority opinion as dominant. This can severely distort collective behavior and decision-making. We study the fundamental question of \emph{detecting} whether a social network allows for a majority illusion. Formally, in the \probName{$q$-Majority Illusion} problem, we ask whether there exists a binary labeling of agents in which at least a~$q$-fraction of agents have the majority of neighbors with the minority label. We investigate how various structural properties of the underlying social network influence the tractability of this question, and provide a detailed map of its computational complexity.
\end{abstract}

\section{Introduction}

It is well-known in social sciences and behavioral economics literature that people tend to adapt their preferences to preferences of their ``neighbors''~\cite{Asch1951,Banerjee1992,CialdiniG2004}. Nowadays, this effect is drastically amplified by social networks. According to a recent survey by \citet{UNESCO2023}, over 56\% of people worldwide consider social media as their primary source of information. This means that their opinions, beliefs, and worldviews are mostly formed on social media.  However, social networks are well known to suffer from various undesirable phenomena, from echo chambers~\cite{GarimellaMGM2018,CinelliMGQS2021}, through polarization~\cite{MuscoMT2018}, misinformation spread~\cite{VosoughiRA2018}, to fake accounts~\cite{FerraraVDMF2016,YangVHM2020}, and many others~\cite{WoolleyH2018}.

An important example of such phenomena is the so-called \emph{majority illusion}. Intuitively, an agent is under majority illusion if they wrongly perceive the minority viewpoint as prevailing in society. More formally, \citet{LermanYW2016} suggest representing the social network as an undirected graph~$G$ where vertices correspond to agents and there is an edge between two agents if they have some social relationship---they are, e.g., friends, colleagues, or neighbors. Moreover, there is a labeling $f$ that assigns each agent one of two possible viewpoints, say \blue{} and \red. Assume that the majority of agents are labeled \blue. Then an agent~$a$ is under majority illusion if the strict majority of their neighbors in $G$ are labeled \red. See \Cref{fig:example} for an instance where \emph{every} agent is under majority illusion.%

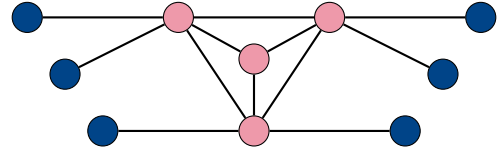
\begin{figure}
    \centering
    \begin{tikzpicture}[every node/.style={draw,circle,inner sep=4pt}]
        \node[fill=red] (v0) at ( 0,0.95) {};
        \node[fill=red] (v1) at ( 0,   0) {};
        \node[fill=red] (v2) at ( 1, 1.5) {};
        \node[fill=red] (v3) at (-1, 1.5) {};

        \node[fill=blue] (v4) at (   2,   0) {};
        \node[fill=blue] (v5) at (   3, 1.5) {};
        \node[fill=blue] (v6) at (-  3, 1.5) {};
        \node[fill=blue] (v7) at (-  2,   0) {};
        \node[fill=blue] (v8) at (-2.5,0.75) {};
        \node[fill=blue] (v9) at ( 2.5,0.75) {};

        \draw[thick] (v1) -- (v2) -- (v3) -- (v1);
        \draw[thick] (v1) edge (v4) edge (v7);
        \draw[thick] (v2) edge (v5) edge (v9);
        \draw[thick] (v3) edge (v6) edge (v8);
        \draw[thick] (v0) edge (v1) edge (v2) edge (v3);
    \end{tikzpicture}
    \caption{An example of a social network where all agents are under majority illusion. There are $4$ red (light) agents and $6$ blue (dark) agents, so blue is the majority label. Every blue agent has exactly one neighbor, which is red. Every red agent has exactly three red neighbors and at most two blue neighbors. Thus, each agent incorrectly perceives the majority label to be red.}
    \label{fig:example}
\end{figure}

Building on the work of \citet{LermanYW2016}, who mainly focused on computational simulations to observe how common majority illusion is in random graphs, \citet{GrandiKLRT2025} initiated the formal study of the computational complexity of detecting majority illusion in social networks. In the \probName{$q$-Majority Illusion} problem, we are given a social network, and the task is to decide whether a labeling~$f$ exists such that at least a $q$-fraction of agents are under majority illusion. In particular, they showed that detecting majority illusion is computationally intractable. Later, \citet{VenemaLosCG2025} extended the work of both \citet{LermanYW2016} and \citet{GrandiKLRT2025} by studying, both theoretically and experimentally, which graph classes guarantee that at least a~$q$-fraction of agents are under majority 
illusion.

We now describe the results of \citet{GrandiKLRT2025} in more detail, as our work directly builds on theirs. They showed that \qMI is \NPc in general, even if the underlying social network is bipartite, of maximum degree~$6$, or planar. In contrast, they proved that the problem is in \XP when parameterized by the treewidth of $G$, and provided \FPT algorithms for parameterization by the vertex cover number and neighborhood diversity. While these results give us a basic understanding of structural restrictions that render the problem (in)tractable, they do not resolve the crucial question: %
    \emph{What structural restrictions make the majority illusion easy to detect?}
And indeed, \citet{GrandiKLRT2025} specifically call for ``another good parametrization.''

\subsection{Our Contribution}

\newcommand{\tabcite}[2][T]{\scriptsize [#1\ref{#2}]}
\newcommand{\tabcitepair}[2]{\scriptsize [T\ref{#1},T\ref{#2}]}

\begin{figure}
    \centering
    \definecolor{myGreen}{HTML}{A2F49B}
    \definecolor{myOrange}{HTML}{FEC44F}
    \definecolor{myYellow}{HTML}{FFF7BC}

    \begin{tikzpicture}[
    node distance=1.2cm and 1.4cm,
    every node/.style={
        draw,
        rounded corners=3pt,
        minimum width=1.1cm,
        minimum height=0.7cm,
        font=\small\ttfamily,
        line width=0.8pt
    },
    FPT/.style={fill=myGreen, draw=myGreen!80!black},
    Wh/.style={fill=myOrange, draw=myOrange!80!black},
    XP/.style={fill=myYellow, draw=myYellow!80!black},
    improved/.style={dashed, line width=1.5pt,draw=black},
    result/.style={line width=1.5pt,draw=black},
    arr/.style={-{Stealth[scale=0.8]}, thick, gray!70}
]
        
        \node[FPT,improved] (vc) at (0,1.7) {vc\,\tabcite[C]{cor:MI:vc:runtime}};
        \node[FPT] (mln) at (-2.9,1.7) {$\phantom{I}$ mln $\phantom{I}$};
        \node[FPT,result] (cdn) at (1.7,1.7) {cdn\,\tabcite{thm:MI:FPT:cdn}};
        \node[FPT] (d2cq) at (3.3,1.7) {$\phantom{I}$dcq$\phantom{I}$};

        \node[FPT,result] (fes) at (-3.8,-.6) {fes\,\tabcite{thm:MI:FPT:fes}};
        \node[XP,Wh,result] (ddp) at (-2,-0.6) {ddp\,\tabcite{thm:MI:W1hard:ddp}};
        \node[FPT,result] (3pvc) at (-0,0.5) {3pvc};
        \node[XP] (tc) at (1.85,-0.6) {tc};
        \node[FPT] (nd) at (3.4,-0.6) {nd};

        \node[XP,Wh,result] (dds) at (-0.3,-0.6) 
        {dds\,\tabcite{thm:MI:W1hard:dds}};

        \node[XP,Wh,result] (4pvc) at (-0.7,-1.85) {4pvc};
        \node[FPT,result] (vi) at (0.7,-1.85) 
        {vi\,\tabcite{thm:MI:FPT:vi}};

        \node[XP,Wh,result] (fvs) at (-2.5,-2.5) {fvs};
        \node[XP,Wh,result] (td) at (0,-2.9) {td\,\tabcite{thm:MI:W1hard:td}};
        \node[XP,Wh,result] (d2cl) at (2.5,-2.5) {dcg\,\tabcitepair{thm:MI:XP:dcg}{thm:MI:W1hard:dcg} };

        \node[XP,Wh] (tw) at (0,-4) {tw};

        \draw[->] (vc) -- (3pvc);
        \draw[->,bend left=20] (3pvc) to (vi);
        \draw[->] (3pvc) -- (dds);
        \draw[->] (dds) -- (4pvc);
        \draw[->] (dds) -- (fvs);
        \draw[->] (3pvc) -- (ddp);
        \draw[->] (vc) -- (tc);
        \draw[->] (vc) -- (nd);
        \draw[->] (mln) -- (fes);
        \draw[->] (mln) -- (ddp);
        \draw[->] (d2cq) -- (tc);
        \draw[->] (d2cq) -- (nd);
        \draw[->] (cdn) -- (tc);
        
        \draw[->] (fes) -- (fvs);
        \draw[->] (ddp) -- (fvs);
        \draw[->] (4pvc) -- (td);
        \draw[->] (vi) -- (td);
        \draw[->] (tc) -- (d2cl);

        \draw[->] (fvs) -- (tw);
        \draw[->] (td) -- (tw);
    \end{tikzpicture}
    \caption{An overview of our results. \probName{$q$-Majority Illusion} is in \FPT for parameters shown in green, while it is \Wh and in \XP for orange ones. A yellow background indicates parameters that are only known to be in \XP without matching \Wh{}ness. An arrow from parameter~$p$ to parameter~$p'$ means that $p'$ is bounded by a function of~$p$, and thus fixed-parameter tractability for~$p'$ implies the same for~$p$. Solid borders indicate results proved in this paper; dashed borders indicate improved running times over previous algorithms. All parameters are formally defined in \Cref{sec:prelim}.} %
    \label{fig:summary}
\end{figure}
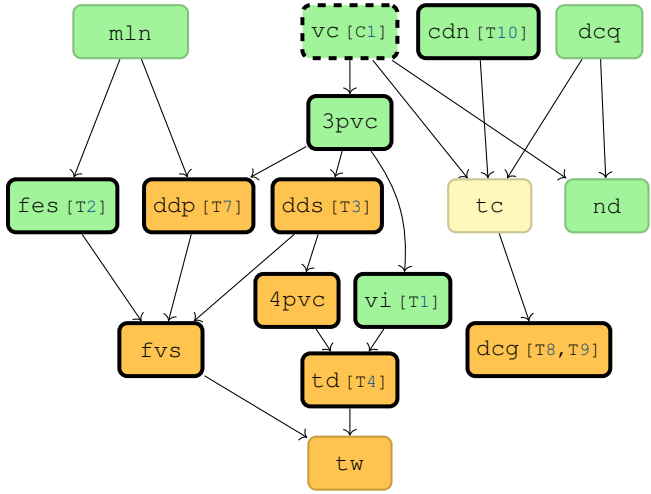

We address the challenge posed by \citet{GrandiKLRT2025} and investigate \emph{what} makes majority illusion tractable. To do so, we systematically study the problem from the perspective of structural graph parameters, significantly improving our understanding of the tractability landscape; see \Cref{fig:summary} for a basic overview of our results. Such approach is very common in computational social choice and social networks analysis; see, e.g., the survey by~\citet{BredereckCW2013}.

We begin with parameterization by vertex integrity. This parameter lies between vertex cover number, which yields an \FPT algorithm, and treewidth, which is known to be in \XP. Informally, vertex integrity measures the number of vertices one must remove to obtain disjoint components of bounded size. This parameter is expected to be bounded for social networks containing a group of ``super-stars'' while the remaining agents form small and independent communities. We show that this parameterization also admits an \FPT algorithm by formulating our problem as an $N$-fold integer linear program---a technique that may be of independent interest. Our algorithm also yields an exponential speedup for parameterization by vertex cover number over the approach of \citet{GrandiKLRT2025}.

Next, we turn our attention to parameters based on small edge- or vertex-distance to some graph class for which computing majority illusion is polynomial-time solvable. First, we show that if the social network has small edge-distance from being a tree (i.e., has small feedback edge set number), then an \FPT algorithm is possible. However, perhaps surprisingly, this approach does not generalize to graphs with small vertex-distance to a tree (i.e., the feedback vertex set). In fact, we show even stronger intractability results. Namely, both the vertex-distance to disjoint paths and the vertex-distance to disjoint stars parameterizations are \Wh. These two results also resolve the probably most important open question left by \citet{GrandiKLRT2025} about the parameterization by treewidth: the above results imply that for this parameterization, an \FPT algorithm is, under standard complexity-theoretic assumptions, unlikely.

Finally, we turn our attention to dense social networks, i.e., networks that are highly connected. Intuitively, one could expect that for such networks, the detection of majority illusion will be much easier from the computational complexity perspective, as the agents have knowledge about the labels of a significant fraction of the other agents. And indeed, \citeauthor{GrandiKLRT2025}~\shortcite{GrandiKLRT2025} showed that if the network is of bounded neighborhood diversity, majority illusion detection is fixed-parameter tractable. We complement their results by showing that for the parameterization by the vertex distance to disjoint cliques, the problem is in \XP and, at the same time, is \Wh. We complement this result with an \FPT algorithm for social networks of bounded \emph{edge} distance to disjoint cliques.

\subsection{Related Work}

Apart from works directly related to the detection of majority illusion described above~\cite{LermanYW2016,GrandiKLRT2025,VenemaLosCG2025}, there are several papers studying the elimination of majority illusion. In this line of research, initiated by \citet{GrandiKLRT2025}, we are given a social network together with a majority-illusion labeling, and the goal is to alter a small number of agents~\cite{FioravantesLLSSW2025}, edges~\cite{GrandiKLRT2025,DippelTNRV2025}, or labels~\cite{JanaR2026} so that the majority illusion is completely eliminated. All of these works use the framework of parameterized complexity to characterize social networks that allow for efficient illusion elimination. However, we note that illusion elimination is substantially different from the detection problem we study in this work, so the results do not carry over to our setting.

Majority illusion is also closely related to models of social influence in social networks~\cite{KempeKT2015,Grandi2017,DoucetteTHLC2019,AulettaFG2020}. In these works, agents have some initial belief, and, in discrete steps, their beliefs update based on the majority belief in their neighborhood. Many works study the computational complexity of deciding how to initially label the social network so that the spread of the preferred opinion is maximized~\cite{KempeKT2015,BredereckE2017,WilderV2018,FaliszewskiGKT2018,BredereckJK2020,CastiglioniFGL2021,Schierreich2023,KnopSS2026}.
Clearly, if one can find a labeling such that each agent is under majority illusion, then in one step, all agents hold the minority opinion.

\section{Preliminaries}
\label{sec:prelim}
Let $G$ be a social network represented as an undirected graph over a set~$\agents$ of agents.
For an agent~$a \in \agents$, let $\neighb_G(a)$ denote the set of its neighbors in~$G$; we let $\deg_G(a)=|\neighb_G(a)|$. 
If $G$ is clear from the context, we may drop the subscript. 
Additionally, we let $\neighb_S(a)=S \cap \neighb(a)$ for any subset~$S$ of agents.
We say that two agents $a,a' \in \agents$ are \emph{twins}
if $\neighb(a) \setminus \{a'\}=\neighb(a') \setminus \{a\}$.

In this paper, a \emph{\labeling} of a set~$A \subseteq N$ of agents is a mapping from~$A$ to $\{\red,\blue\}$, classifying each agent in~$A$ as either red or blue; a labeling for~$G$ is one for~$\agents$. 
For a set~$A$ of agents and a \labeling~$
\rblab$ of~$G$, let $\RedNo_\rblab(A)$  and $\BlueNo_\rblab(A)$ denote the number of red and blue agents within~$A$, resp. 
We will usually assume that blue is the 
\emph{majority color in~$\rblab$}, that is, $\BlueNo_\rblab(\agents)>\RedNo_\rblab(\agents)$.
Some agent~$a$ has \emph{majority-red neighborhood} if $\RedNo_\rblab(\neighb(a))>\BlueNo_\rblab(\neighb(a))$, and 
$a$ is \emph{under majority illusion} if the majority color for~$\rblab$ is blue but $a$ has majority-red neighborhood (or vice versa).
For a set~$A \subseteq \agents$ of agents, we let 
$\umi(A)$ denote the number of agents under majority illusion in~$A$ under~$\rblab$, and 
$\mrn(A)$ the number of agents in~$A$ with a majority-red neighborhood. 

\subsection{Structural Graph Parameters}

Let us briefly summarize the key graph parameters that we use in our study to describe the structural properties of a social network. 

Given a graph~$G=(V,E)$ and a set $S$ of vertices (or edges), let $G-S$ denote the graph obtained from~$G$ by deleting all vertices (or edges, resp.) of~$S$.
The \emph{feedback vertex} (or \emph{edge}) \emph{set number} of~$G$, denoted by $\fvs(G)$ (or $\fes(G)$), is the minimum size of a set~$S$ of vertices (or edges, resp.) such that $G-S$ is acyclic. The \emph{distance to cluster graphs, to disjoint paths,} or \emph{to disjoint stars} of~$G$, denoted by $\dcg(G)$, $\ddp(G)$, or $\dds(G)$, is the minimum size of a vertex set~$S$ such that $G-S$ is the disjoint union of cliques, paths, or stars, respectively. The \emph{twin-cover} ($\operatorname{tc}$) 
of~$G$ is a variant of the distance to cluster graph where all vertices of $G-S$ are twins in $G$.
An edge equivalent of $\dcg(G)$ is \emph{cluster edge deletion number} $\cdn(G)$, which is the minimum number of edges we need to remove from $G$ to obtain a disjoint union of cliques.

The \emph{vertex cover number} of~$G$, denoted by~$\vc(G)$ is the minimum size of set~$S \subseteq V$ such that $G-S$ has no edges.
The \emph{vertex integrity} of~$G$, denoted by $\vi(G)$, is the value $\min_{S \subseteq \agents}|S|+\maxcc(G-S)$ where $\maxcc(G-S)$ is the size of a largest connected component in $G-S$; a \emph{$\vi$-set} of~$G$ is then a set~$S$ for which this expression takes the minimum.
The \emph{maximum leaf number} ($\operatorname{mln}$) of a graph~$G$ is the maximum number of leaves in a spanning tree of~$G$.
The \emph{distance to clique} %
of~$G$ is the minimum number of vertices whose deletion from~$G$ results in a single clique. 
The \emph{$\ell$-path vertex cover} %
of~$G$ for some integer~$\ell$ is the minimum number of vertices whose deletion from~$G$ results in a graph in which no simple path contains $\ell$ vertices.
The \emph{neighborhood diversity} %
of~$G$ is the minimum number of vertex sets into which the vertex set of~$G$ can be partitioned such that all vertices within each set are twins.

The \emph{span} of a rooted tree is the graph constructed from that tree by adding
an edge $\{u,v\}$ for each pair of vertices $u$ and~$v$ where $u$ is an ancestor of~$v$. The \emph{treedepth} of~$G$ is then the smallest height of a rooted tree~$T$ such that $G$ is a
subgraph of the span of~$T$.

A \emph{tree-decomposition} of~$G$ is a pair $(T,(B_t)_{t \in V(T)})$ where $T$ is a tree whose each node~$t$ is associated with a subset $B_t$ of vertices in~$G$, called a \emph{bag}, in a way that (i) each vertex of~$G$ appears in some bag, (ii) for each edge in~$G$ there is a bag containing both of its endpoints, and 
(iii) for each vertex~$v$ of~$G$, the set $\{t \in V(T): v \in B_t\}$ induces a subtree of~$T$.
The \emph{width} of a tree-composition is the size of its largest bag minus~$1$, 
and the \emph{treewidth} of~$G$ is the minimum width of any tree-decomposition of~$G$.

\subsection{Parameterized Complexity}

The framework of parameterized complexity allows us to examine computationally hard problems in a detailed fashion.
In a parameterized problem~$Q$, each input instance~$I$ is associated with an integer parameter~$k$, and an algorithm's running time is measured as a function of both~$|I|$ and~$k$. 
An algorithm for~$Q$ is \emph{fixed-parameter tractable} (\FPT) if it runs in~$f(k)|I|^{\Oh{1}}$ for some computable function~$f$. 
By contrast, an $\XP$ algorithm runs in~$|I|^{\Oh{f(k)}}$ time for some computable function~$f$. 
The class of parameterized problems admitting an \FPT (or \XP) algorithm is denoted by \FPT (or \XP, respectively). 
Showing that some problem is \emph{not} likely to be in \FPT can be done by proving its \Whness, via a parameterized reduction from some parameterized problem already known to be \Wh. 
For an introduction and more details on this framework, see e.g., the book~\cite{cygan2015parameterized}.

\subsection{$\boldsymbol{N}$-fold Integer Programming}

In an $N$-fold IP, the goal is to minimize a linear objective $f$ over a set of structured constraints. Formally, let $r,D\in\N$ and $s_i,t_i\in \N$ for every $i\in[D]$. An $N$-fold IP contains $d = \sum_{i\in[D]} t_i$ variables partitioned into $D$ \emph{bricks}; let vector $\myvec{x^{(i)}}$ denote be the $i$-th brick. Then, the constraints %
have the following form:
\begin{align}
    && \mymatrix{E_1}\myvec{x^{(1)}} + \cdots + \mymatrix{E_D}\myvec{x^{(D)}} &= \myvec{b}_0 && \label{def:Nfold:global} \\
    \forall i \in [D] && \myvec{A_i} \myvec{x^{(i)}} &= \myvec{b}_i && \label{def:Nfold:local} \\
    \forall i \in [D] && \boldsymbol{\ell}_i \leq \myvec{x^{(i)}} &\leq \myvec{u}_i \label{def:Nfold:box} &&
\end{align}
where $\mymatrix{E_i} \in \Z^{r\times t_i}$, $\mymatrix{A_i} \in \Z^{s_i\times t_i}$, $\myvec{b}_0 \in \Z^r$, $\myvec{b}_i \in \Z^{s_i}$, and $\myvec{\ell}_i,\myvec{u}_i\in\Z^{t_i}$ for each $i\in[D]$. We call constraints of type \eqref{def:Nfold:global} \emph{global} and constraints of type \eqref{def:Nfold:local} \emph{local}.

\section{Bounded Vertex Integrity}

We begin by showing the fixed-parameter tractability of \probName{$q$-Majority Illusion}  with respect to the vertex integrity of the social network. This parameter is indeed relevant for social networks, as there are examples where we have a set of ``influencers'' and the remaining agents can be partitioned into small and disjoint communities. Also note that this parameter has been used in several similar works~\cite{DvorakEGKO2017,BodlaenderHJOOZ2020,FioravantesGM2025,KnopSS2026}. The algorithm is based on an $N$-fold ILP formulation of the \probName{$q$-Majority Illusion} problem.

\begin{theorem}\label{thm:MI:FPT:vi}
    \probName{$q$-Majority Illusion} is in \FPT %
    when parameterized by the vertex integrity $\operatorname{vi}(G)$ of the social network~$G$.
\end{theorem}
\begin{proof}
    Let $\mathcal{I}=(N,G,q)$ be an instance of \probName{$q$-Majority Illusion} and $M$ be a vi-set for~$G$ of size $\operatorname{vi}(G) = k$. Our algorithm can be divided into two parts. First, we guess how a hypothetical solution \labeling~$\rblab$ labels the agents of $M$, and which agents of $M$ are under illusion in $f$ (the latter can be skipped if $q=1$, as then $\umi(\agents)=|\agents|$).
    Then, for each such possibility, we create an $N$-fold ILP that verifies whether we can extend the guessed partial labeling of~$M$ so that the correct number of agents actually experience majority illusion.

    We now formally describe the verification phase of the algorithm. Let $\rblab_M$ be some \labeling\ of the vertices of $M$ and $M'\subseteq M$ be a set of agents of $M$ under majority illusion in a hypothetical solution $\rblab$ (in an actual solution, more agents of $M$ can be under majority illusion, but the program we construct checks this only for agents of~$M'$).
    
    Let $\mathcal{H}$ %
    be the set of connected components of $G-M$.
    By definition, each component is of size at most $k$; therefore, there are at most $2^k$ different labelings%
    for every component $H\in\mathcal{H}$. We call each such labeling a \emph{configuration} of~$H$, and we use $\mathcal{C}_H$ to denote the set of all possible configurations for a component $H$. Additionally, we denote by $\mu_{H,C}$ the number of agents in component~$H$ who are under majority illusion if $H$ is in (i.e., labeled by) configuration $C\in\mathcal{C}_H$, assuming that we use labeling $f_M$ for agents in~$M$ (recall that $f_M$ is known and fixed at this time). Furthermore, we use $\BlueNo_{H,C}$ to denote the number of blue agents and $\RedNo_{H,C}$ to denote the number of red agents, respectively, in~$H$ when it is in configuration~$C$. Finally, we extend the previous notion and use $\BlueNo_{H,C}(\neighb_H(v))$ ($\RedNo_{H,C}(\neighb_H(v))$) to denote the number of blue (red, respectively) neighbors of some agent $v\in M$ in~$H$ when $H$ is in configuration~$C$. 

    Then, our $N$-fold ILP contains a binary variable $x_{H,C}$ for every $H\in\mathcal{H}$ and every $C\in\mathcal{C}_H$. If this variable is set to~$1$ for some $H$ and $C$, then, in the corresponding solution, the component~$H$ is in configuration~$C$. The program contains four types of constraints.

    First, we have a single global constraint that ensures that blue is indeed the majority color: 
    \begin{equation}
        \sum_{H\in\mathcal{H}} \sum_{C\in\mathcal{C}_H} x_{H,C}\cdot \BlueNo_{H,C} > \frac{|N|}{2} - \BlueNo_{f_M}(M)\,.\label{eq:MI:FPT:vi:blueIsMajority}
    \end{equation}

    Second, we have $|M'|$-many constraints, one for each agent of $M'$, which ensure that all agents of $M'$ are indeed under majority illusion.
    \begin{multline}
        \forall v\in M'\colon \sum_{H\in\mathcal{H}}\sum_{C\in\mathcal{C}_H}  x_{H,C}\cdot \RedNo_{H,C}(\neighb_H(v))  \\ 
         >
        \frac{\deg_G(v)}{2} - \RedNo_{f_M}(\neighb_M(v))\,,        
    \label{eq:MI:FPT:vi:modulatorAgentsUnderIllusion}
    \end{multline}
    and for agents of $M\setminus M'$, we (for technical reasons) require that none of them is under majority illusion.
    \begin{multline}
        \forall v\in M\setminus M'\colon \sum_{H\in\mathcal{H}}\sum_{C\in\mathcal{C}_H}  x_{H,C}\cdot \RedNo_{H,C}(\neighb_H(v))  \\ 
         \leq
        \frac{\deg_G(v)}{2} - \RedNo_{f_M}(\neighb_M(v))\,,        
            \label{eq:MI:FPT:vi:modulatorAgentsWithoutIllusion}
    \end{multline}

    Next, we add a single global constraint that ensures that enough agents are under majority illusion.
    \begin{equation}
        \sum_{H\in\mathcal{H}}\sum_{C\in\mathcal{C}_H} x_{H,C}\cdot \mu_{H,C} \geq \lceil q\cdot|N| \rceil  - |M'|\,.\label{eq:MI:FPT:vi:enoughAgentsUnderIllusion}
    \end{equation}
    Notice that if we remove the previous constraint and instead maximize the left side of the inequality in the objective function, we can easily find an extension of $f_M$ with the largest number of agents under majority illusion. However, since we formulate $\mathcal{I}$ as a decision problem, we keep the constraint as is and use a constant objective function.

    Finally, we need to ensure that each component $H \in \mathcal{H}$ is in exactly one configuration. This can be secured by adding a single local constraint for each  $H\in \mathcal{H}$ of the following form.
    \begin{equation}
        \forall H\in\mathcal{H}\colon \sum_{C\in\mathcal{C}_H} x_{H,C} = 1\,.\label{eq:MI:FPT:vi:eachComponentOneConfig}
    \end{equation}

    Now, we prove that it is possible to extend the given labeling~$f_M$ to a solution $f$ for $\mathcal{I}$ such that agents of $M'$ are under majority illusion if and only if the above $N$-fold IP is feasible. We start with the left-to-right implication.

    \begin{claim}
    \label{clm:nfold1}
        Let $f$ be an extension of $f_M$ so that at least $q\cdot |N|$ agents, including agents of $M'$ and excluding agents of $M\setminus M'$, are under majority illusion. Then, the constructed $N$-fold ILP admits a feasible solution~$\myvec{x}$.
    \end{claim}
    \begin{claimproof}
        Without loss of generality, we assume that blue is the predominant color. If not, we can simply swap the color of each agent and obtain another solution. Now, for every component $H\in\mathcal{H}$, we set $x_{H,C} = 1$ if and only if $H$ is in configuration~$C$ according to $f$. Clearly, for every $H$, there exists exactly one $C$ such that $x_{H,C} = 1$, so constraints of type \eqref{eq:MI:FPT:vi:eachComponentOneConfig} are easily satisfied.

        Let $\BlueNo_f$ denote the number of blue agents according to $f$. By our assumption, $\BlueNo_f > |N|/2$. Moreover, since $f$ is an extension of $f_M$, it holds that $\BlueNo_f - \BlueNo_{f_M} > |N|/2 - \BlueNo_{f_M}$. Observe that the left part of the inequality corresponds to the number of blue agents in components $H_1,\ldots,H_\ell$, so we can write it as $\sum_{H\in \mathcal{H}} \BlueNo_{H,f}$, where $\BlueNo_{H,f}$ is the number of agents in~$H$ colored blue in~$f$. This is exactly 
        \[\sum_{H\in\mathcal{H}} x_{H,C}\cdot\BlueNo_{H,C} = \sum_{H\in\mathcal{H}}\sum_{C\in\mathcal{C}_H} x_{H,X} \cdot \BlueNo_{H,C},
        \] 
        since exactly one variable~$x_{H,C}$ for some fixed component~$H$---the one corresponding to the configuration of $H$ in $f$---is set to~$1$. That is, constraint \eqref{eq:MI:FPT:vi:blueIsMajority} is satisfied.

        Let $v\in M'$ be an agent. By our assumptions, $v$ must be under majority illusion in $f$. That is, the number of red agents in its neighborhood is greater than  $\deg_G(v)/2$. Since $f$ is an extension of $f_M$, the number of red neighbors of $v$ in the components $H_1,\ldots,H_\ell$ is greater than $\deg_G(v)/2 - \RedNo^v_{f_M}$. Using an argument similar to that in the previous paragraph, we obtain that constraint~\eqref{eq:MI:FPT:vi:modulatorAgentsUnderIllusion} is satisfied for $v$. Since we picked~$v$ arbitrarily, all these constraints are necessarily satisfied. Similarly, constraint~\eqref{eq:MI:FPT:vi:modulatorAgentsWithoutIllusion} is satisfied for all $v\in M\setminus M'$.

        Finally, since $f$ is an extension of $f_M$ with at least $q\cdot |N|$ agents under majority illusion, including the agents of $M'$ and excluding agents of $M\setminus M'$, there are at least $q\cdot|N|-|M'|$ agents within the components $H_1,\ldots,H_\ell$ under illusion. 
        The number of such agents is $\sum_{H\in\mathcal{H}}\sum_{C\in\mathcal{C}_H} x_{H,C}\cdot\mu_{H,C}$ which is an integer at least $q\cdot|N|-|M'|$;  therefore constraint \eqref{eq:MI:FPT:vi:enoughAgentsUnderIllusion} is satisfied as well.
    \end{claimproof}

    Now, we show that each feasible solution $\myvec{x}$ for the $N$-fold ILP corresponds to a solution extension $f$ of $f_M$.

    \begin{claim}
    \label{clm:nfold2}
        Let $\myvec{x}$ be a feasible solution for the $N$-fold ILP. Then, there is an extension of $f$ such that at least $q\cdot |N|$ agents, including agents of $M'$ and excluding agents of $M\setminus M'$, are under majority illusion.
    \end{claim}
    \begin{claimproof}
        Given a feasible solution $\myvec{x}$, we construct the extension~$f$ of~$f_M$ as follows. For every component $H\in\mathcal{H}$ and every vertex $v\in V(H)$, we set $f(v) = C(v)$, where $C\in\mathcal{C}_H$ is the unique configuration of $H$ such that $x_{H,C} = 1$. By constraint~\eqref{eq:MI:FPT:vi:eachComponentOneConfig}, this labeling $f$ is well defined, since for each $H \in \mathcal{H}$ there is exactly one variable $x_{H,C}$ set to~$1$. Now, the total number of agents colored blue is 
        \begin{align*}
        \BlueNo_f &= \BlueNo_{f_M} + \sum_{H\in\mathcal{H}} \BlueNo_{H,f} 
        \\
        &= \BlueNo_{f_M} + \sum_{H\in\mathcal{H}}\sum_{c\in\mathcal{C}_H} x_{H,C}\cdot \BlueNo_{H,C}  >|N|/2  
        \end{align*}
        where $\BlueNo_{H,f}$ is again the number of agents in~$H$ colored blue in~$f$, 
        and the inequality follows from 
        constraint~\eqref{eq:MI:FPT:vi:blueIsMajority}. That is, the blue color is indeed the majority color.

        What remains to verify is that the correct number of agents, including all agents of $M'$, is experiencing a majority illusion. Let $v\in M'$ be an agent. If $\sum_{H\in\mathcal{H}} \RedNo^v_{H,f} + \RedNo_{f_M}^v \leq \deg_G(v)/2$, then %
        constraint~\eqref{eq:MI:FPT:vi:modulatorAgentsUnderIllusion} is violated for $v$, which contradicts that $\myvec{x}$ is a feasible solution. That is, all agents in~$M'$ are indeed under the illusion in~$f$. By the same argument, no agent of $M\setminus M'$ is under majority illusion, as otherwise, constraint~\eqref{eq:MI:FPT:vi:modulatorAgentsWithoutIllusion} is violated. Finally, the number of agents in components of $\mathcal{H}$ under illusion can be expressed as $\sum_{H\in\mathcal{H}} \mu_{H,f}$, which is equal to $\sum_{H\in\mathcal{H}}\sum_{C\in\mathcal{C}_H} x_{H,C}\cdot\mu_{H,C}$. By constraint \eqref{eq:MI:FPT:vi:enoughAgentsUnderIllusion}, this is at least $q\cdot|N| - |M'|$. Additionally, all agents of $M'$ are under the illusion and, therefore, at least $q\cdot |N| - |M'| + |M'|$ are under majority illusion, which finishes the proof.
    \end{claimproof}

    Let us prove the correctness of the algorithm
    using \Cref{clm:nfold1,clm:nfold2}. %
    Our algorithm tries all possible labelings $f_M$ over~$M$ and, for each, all possible sets~$M'$; thus, if there is a solution $f$ for~$I$, then the algorithm checks its restriction to~$M$ and, by \Cref{clm:nfold1}, returns the correct answer. If there is no solution for~$I$, the $N$-fold ILP is not feasible for any combination of $f_M$ and $M'$, so it correctly returns \emph{no}.

    For running time, the $N$-fold ILP contains $k+2$ global constraints and one constraint that is local for every brick. The largest coefficient in the constraints is also $k$, so we can solve the previous ILP in time $2^{\Oh{k^2\log k}}\cdot n^\Oh{1}$ time, using the algorithm of \citet{EisenbrandHKKLO2025}.

    For the guessing part, there are $2^\Oh{k}$ different functions~$f_M$ and, for each such function, $2^\Oh{k}$ different sets~$M'$. As we construct the previous $N$-fold ILP for every pair of possible $f_M$ and $M'$, we obtain an algorithm running in $2^\Oh{k}\cdot 2^\Oh{k} \cdot 2^\Oh{k^2\log k}\cdot n^\Oh{1} \in 2^\Oh{k^2\log k}\cdot n^\Oh{1}$, which is indeed in \FPT. Therefore, the theorem follows.
\end{proof}

Using the previous algorithm, we also significantly improve the running time for parameterization by the vertex cover number. The best known algorithm so far, %
due to \citeauthor{GrandiKLRT2025}~\shortcite{GrandiKLRT2025}, runs in $2^{2^\Oh{\operatorname{vc}(G)}}|\agents|^\Oh{1}$ time, i.e., is doubly exponential in the parameter. As $\vi(G) \leq \vc(G)+1$, by employing the previous algorithm, we directly obtain the following.

\begin{corollary}
\label{cor:MI:vc:runtime}
    The \probName{$q$-Majority Illusion} problem 
    can be solved in $2^\Oh{\operatorname{vc}(G)^2\log \operatorname{vc}(G)}\cdot n^\Oh{1}$ time.
\end{corollary}

\section{Bounded Feedback Edge Set}

Recall that the feedback edge set (FES) is a set of edges whose removal turns the social network into a tree. This parameter is particularly appealing because, unlike other parameters we study in this work, it can be computed in polynomial time. FES found use in the design of practical algorithms in various AI applications~\cite{GruttemeierK2020,GanianK2021,GuoLN0N2022}. 

Our $\FPT$ algorithm for FES relies on a polynomial-time algorithm for the following auxiliary problem.

\begin{definition}
    The input of the \ListMI problem is a triple $(\agents,G,L)$ where $\agents$ is a set of agents, $G$ is a social network over~$\agents$, and $L$ is a function $L\colon\agents \to \{\{\blue\},\{\red\},\{\blue,\red\}\}$ %
    that satisfies $|\{v \in \agents : \blue\in L(v)\}| > |\agents|/2$. 
    The task is to compute the maximum number of agents under majority illusion over all possible labelings~$f$ where
    \begin{itemize}
        \item $f(v) \in L(v)$ for every $v\in \agents$, and 
        \item $\BlueNo_f(\agents) > \RedNo_f(\agents)$, i.e., blue is the majority color.
    \end{itemize}
\end{definition}

Observe that due to our condition on $L$, each instance admits at least one  labeling~$f$ satisfying the above conditions. We show that if $G$ is a tree, then we can compute the solution for it in polynomial time.

\begin{lemma}
\label{lem:LMRMI:poly:tree}
    If the social network $G$ is a tree, \ListMI is solvable in polynomial time.
\end{lemma}

\ifCutLemmaLMRMIpolytree
\begin{proofsketch}
    \TODOnote{Write this.}
\end{proofsketch}

\else
\begin{proof}
    We use dynamic programming over the tree~$G$. Let $\agents$ be the set of agents in~$G$. Consider $G$ as a rooted tree with an arbitrarily chosen agent~$\root$ as its root. For each agent~$v \in \agents$, let $\insubtree{v}$ denote the set of agents in the subtree of~$G$ rooted at~$v$.
    Moreover, if $v \neq \root$, then we write $\parent(v)$ for the parent of~$v$, and let $\children(v)=\{u \in \agents: \parent(u)=v\}$.

    For each agent $v\in\agents\setminus\{\root\}$, we compute a table~$\DP_v$ containing values $\DP_v[c,c_p,B]$ where
    \begin{itemize}
        \item $c\in\{\red,\blue\}$ is the color of agent $v$,
        \item $c_p\in\{\red,\blue\}$ is the color of $p = \operatorname{parent}(v)$, and
        \item $B\in[|\agents|]$ is the number of agents in~$N^v$ colored $\blue$.
    \end{itemize}
    We call triplets of the form $(c,c_p,B)$ a \emph{signature}, and for each signature, the table stores the maximum number of agents in~$\insubtree{v}$ that are under majority illusion assuming that $v$ is of color $c$, $p$ is of color $c_p$, there are exactly $B$ blue agents in the subtree rooted at~$v$, and the color of each agent $u\in \insubtree{v}$ is in its corresponding set $L(u)$. If no such labeling exists, the table stores $-\infty$. If $\DP_v[c,c_p,B] \not= -\infty$, we say that the signature $(c,c_p,B)$ is \emph{valid} for~$v$.

    Once the dynamic programming table is correctly computed for all agents, we can compute the solution to our instance as follows. Let $u_1,\ldots,u_\gamma$ be an arbitrary but fixed ordering of the children of the root $\root$. We return 
    \begin{equation}
        \max_{c_\rho\in L(\root)}
        \max_{B\in[\lceil|\agents|/2 + 1\rceil, |\agents|]}
        \phantom{xx}
        \max_{\mathclap{\substack{
            \myvec{B}\in[B]^\gamma\\
            \myvec{c}\in\{\red,\blue\}^\gamma\\
            \sum_{i=1}^\gamma \myveccomp{B}_i   = B - \llbracket c_\rho = \blue\rrbracket
        }}}
            \iota(\myvec{c}) + \sum_{i=1}^\gamma \DP_{u_i}[ \myveccomp{c}_i, c_\rho, \myveccomp{B}_i ] \label{eq:LMRMI:poly:tree:finalComputation}
    \end{equation}
    where $\iota$ is a function that evaluates to $1$ if the strict majority of elements of the vector $\myvec{c}$ have the value $\red$, and to $0$ otherwise. Moreover, we use $\llbracket \cdot \rrbracket$ as the Iverson bracket; i.e., an expression that evaluates to $1$ if the condition inside the brackets is satisfied and to $0$ otherwise.
 
    \proofsubparagraph{Computation} Now, we formally show how to compute the dynamic programming table for each agent~$v$ and for every signature $(c,c_p,B)$. We distinguish two cases according to the type of agent~$v$ under consideration.
    
    \smallskip
    \noindent\emph{Leaf Agent.}\hspace{0.2cm} Let $v$ be a leaf agent. Every valid signature must satisfy that the color~$c$ is allowed for~$v$, i.e., $c \in L(v)$. Moreover, the number~$B$ of blue vertices in the subtree is one if and only if the color of~$c$ is blue. Otherwise, it must hold that $B=0$. Once we get rid of all clearly invalid signatures, we have that agent~$v$ is under majority illusion if and only if its parent is red. Formally, the computation is as follows:
    \begin{equation}
    \label{eq:fes:DP-leaf}
        \DP_v[c,c_p,B] = \begin{cases}
            -\infty & \text{if } c \not\in L(v) \text{ or }
            B > 1,\\
            -\infty & \text{if } c = \red \text{ and } B = 1,\\
            -\infty & \text{if } c = \blue \text{ and } B = 0,\\
            1       & \text{if } c_p = \red,\text{ and}\\
            0       & \text{otherwise.}
        \end{cases}
    \end{equation}

    \smallskip
    \noindent\emph{Internal Agent.}\hspace{0.2cm}
    If $\DP_v[c,c_p,B]$ is such that $c \not\in L(v)$, we set $\DP_v[c,c_p,B] = -\infty$. Otherwise, we find an optimal way how to label the children of $v$ and how to split the remaining blue agents between their respective subtrees. Let %
    $\children(v) = \{u_1,\ldots,u_\gamma\}$. The  computation is %
    as follows.
    \begin{equation}
    \label{eq:fes:DP-internal}
        \DP_v[c,c_p,B] = 
            \max\limits_{\mathclap{\substack{
                \myvec{B} \in [B]^{\gamma}\\
                \myvec{c} \in \{\blue,\red\}^{\gamma}\\
                \sum_{i=1}^\gamma \myveccomp{B}_i = B -  \llbracket c = \blue\rrbracket 
            }}}
                \iota(c_p,\myvec{c}) + \sum\limits_{i=1}^{\gamma} \DP_{u_i}[\myveccomp{c}_i,c,\myveccomp{B}_i]\,.
    \end{equation}

    \proofsubparagraph{Correctness}
    We show that, indeed, each dynamic programming table is computed correctly. %

    \begin{claim}\label{clm:LMTRI:DP-correct:ltr}
        Whenever $\DP_v[c,c_p,B] = w$ and $w \not= -\infty$, there exists a labeling $f\colon \insubtree{v}\to\{\blue,\red\}$ such that ${f(v) = c}$, $\BlueNo_f(\insubtree{v}) = B$, $f(u) \in L(u)$ for every $u\in \insubtree{v}$, and, assuming the color of $\operatorname{parent}(v)$ is $c_p$, the number of agents under majority illusion in~$\insubtree{v}$ is exactly $w$.
    \end{claim}
    \begin{claimproof}
        We prove the claim by bottom-up induction over the agents in the tree. Let $v$ be a leaf agent that admits a valid signature, so that  $\DP_v[c,c_p,B] = w \neq -\infty$. By our computation in~\eqref{eq:fes:DP-leaf}, $w\in\{0,1\}$. We define~$f$ by setting $f(v) = c$ and $f(v_p)=c_p$. Clearly, $c\in L(v)$, as otherwise we would have $w = -\infty$ by the first condition in~\eqref{eq:fes:DP-leaf}.
        If $c = \red$, then $B = 0 = \BlueNo_f(\insubtree{v})$, and if $c=\blue$, then $B = 1 = \BlueNo_f(\insubtree{v})$. Finally, the number of agents in~$\insubtree{v}$ under the illusion with majority-red neighborhood, by $f(v_p)=c_p$ is $1$ if and only if $c_p = \red$. Thus, $f$ has the desired properties.

        Assume now that $v$ is an internal agent and the claim holds for all its children; let $u_1,\dots,u_\gamma$ denote them. Let $\myvec{c}$ and~$\myvec{B}$ be the vectors where the right-hand side of~\eqref{eq:fes:DP-internal} takes its maximum. Since $\DP_v(c,c_p,B)\neq - \infty$, we know that for each $i \in [\gamma]$, the signature $(\myveccomp{c}_i,c,\myveccomp{B}_i)$ is valid for~$u_i$.
        for each $i \in [\gamma]$, let $f_i$ be the labeling for $\insubtree{u_i} \cup \{v\}$ guaranteed by the inductive hypothesis. We define~$f$ by taking the union of $f_1,\dots,f_\gamma$ (note that $f_i(v_p)=c_p$ for each $i \in [\gamma]$, so $f$ is well-defined.
        Using again the inductive hypothesis, the number of blue agents under~$f$ in~$\insubtree{v}$ is exactly $(\sum_{i=1}^\gamma \myveccomp{B}_i)+\llbracket c=\blue \rrbracket$, and we also have $f(u) \in L(u)$ for all agents $u \in \insubtree{v}$, using also that $c \in L(v)$ holds as otherwise we would have $w = - \infty$.
        Finally, the number of agents under majority illusion in~$\insubtree{u_i}$ is 
        $\DP_{u_i}(\myveccomp{c}_i,c,\myveccomp{B}_i)$ for each $i \in [\gamma]$ by induction; hence, by the definition of~$\iota$, the right-hand side of~\eqref{eq:fes:DP-internal} correctly computes the number of agents with a majority-red neighborhood in~$\insubtree{v}$ under~$f$.
     \end{claimproof}

    Next, we show that the value stored for each signature is indeed the maximum possible number of agents under illusion.

    \begin{claim}\label{clm:LMTRI:DP-correct:rtl}
        If a labeling $f\colon \insubtree{v} \cup \{v_p\} \to\{\blue,\red\}$ exist such that $f(v) = c$, $f(v_p)=c_p$,  $\BlueNo_f(\insubtree{v}) = B$, and $f(u)\in L(u)$ for every $u\in \insubtree{v}$, then $\DP_v[c,c_p,B] \geq \mrn[f](\insubtree{v})$.
    \end{claim}
    \begin{claimproof}
        Again, we use bottom-up induction over the agents of the tree. First, let $v$ be a leaf agent and $f$ be a labeling corresponding to a signature $(c,c_p,B)$. By the assumptions on $f$, it must hold that $\DP_v[c,c_p,B] \in \{0,1\}$. Moreover, $\DP_v[c,c_p,B] = 1$ if and only if $c_p = \blue$, showing that the claim is indeed satisfied for leaf agents.

        Now, let $v$ be an internal agent and assume that the claim holds for all children $u\in\children(v)= \{u_1,\dots,u_{\gamma}\}$. Let $f_i$ be a restriction of~$f$ to~$\insubtree{u} \cup \{v\}$
        for a child~$u_i$ of~$v$.
        Clearly, $f_i$ corresponds to a signature $(f(u),f(v),\BlueNo_f(\insubtree{u}))$ in $\DP_u$ and, by the induction hypothesis, \[\DP_u[f(u),f(v),\BlueNo_f(\insubtree{u})] \geq \mrn[f^u](\insubtree{u}).\] By~\eqref{eq:fes:DP-internal}, the value stored in $\DP_v[c,c_p,B]$ is at least 
        \begin{align*}
        \sum_{i=1}^{\gamma} &\DP_{u_i}[f(u_i),f(v),\BlueNo_f(\insubtree{u_i})] + \iota(c_p,(f(u_i))_{i\in[\gamma]}) \\
        & \geq \sum_{i=1}^{\gamma} \mrn[f_i](\insubtree{u_i}) + \iota(c_p,(f(u_i))_{i\in[\gamma]}) \geq \mrn[f](\insubtree{v}),
        \end{align*}
        as desired.
    \end{claimproof}
    
    Next, we show that the computation of the final output as formulated in \eqref{eq:LMRMI:poly:tree:finalComputation} is also correct, which leads to the correctness of our approach. %
    According to \Cref{clm:LMTRI:DP-correct:ltr,clm:LMTRI:DP-correct:rtl}, the dynamic programming table is correctly computed for all children $u\in\children(\root)$. Recall also that there is at least one labeling for~$G$ as required and thus an optimal such labeling~$f$ always exists. Note that \eqref{eq:LMRMI:poly:tree:finalComputation} checks the combination of $c = f(\root)$, $B = \BlueNo_f(\insubtree{\rho})$, $\myvec{B} = (\BlueNo_f(\insubtree{u}))_{u\in\children(\root)}$, $\myvec{c} = (f(u))_{u\in\children(\root)}$, and, as the dynamic tables are computed correctly, the algorithm returns value of at least $\mrn[f](\agents)$ which, as the majority color  by~$B=\RedNo_f(\agents) >|\agents|/2$ is blue, equals $\umi[f](\agents)$.
 
    \proofsubparagraph{Running Time}
    The computation in a leaf agent can be done in constant time by simply checking all conditions. To compute a cell's value for an internal agent (as well as the final result in \eqref{eq:LMRMI:poly:tree:finalComputation}), we use simple dynamic programming, which runs in $\Oh{|\agents|^3}$ time. In particular, the dynamic programming table used within an internal agent~$v$  is~$\texttt{T}[i,\rho,\beta]$, where~$i$ is the currently processed child of~$v$,~$\rho$ is the number of children in~$\{u_i,\ldots,u_\gamma\}$ colored red, and~$\beta$ is the number of agents of~$T_{u_i},\ldots,T_{u_\gamma}$ we need to color blue. The stored value is the maximum number of agents under majority illusion in subtrees~$T_{u_i},\ldots,T_{u_\gamma}$, assuming that~$v$ is of color~$c$, exactly~$\rho$ roots of these subtrees are colored red, and there are exactly~$\beta$ blue agents in these subtrees. The size of the table is at most~$n\cdot \gamma\cdot n$, and each cell can be computed in time~$\Oh{|T_{u_i}|}$ by a) trying to color~$u_i$ in color~$c_i\in L(u_i)$, b) guessing the number of blue agents~$B_i$ in the respective subtree~$T_{u_i}$, and c) looking up the optimal solution for~$i+1$ with the parameters~$\rho$ and~$\beta$ appropriately modified (this is already stored in~$\texttt{T}[i+1,\cdot,\cdot]$). Moreover, for each pair of~$c_i$ and~$B_i$, we already have the partial solution computed in table~$\DP_{u_i}$. Observe that overall, the computation can be done in~$\Oh{|\agents|^2\cdot \gamma \cdot |\agents|}$ time. However, as~$\gamma \in \Oh{|\agents|}$, the overall running time of the above approach is~$\Oh{|\agents|^4}$. With a more careful analysis based on the fact that~$\sum_{i\in[\gamma]} |T_i| = \Oh{|\agents|}$, we can implement the DP in~$\Oh{|\agents|^3}$ time. We have $\Oh{|\agents|}$ tables with $\Oh{|\agents|}$ signatures for each table, giving us the overall running time of $\Oh{|\agents|^5}$.
\end{proof}

\fi 
Finally, we are ready to give our algorithm for \qMI. Given a feedback edge set~$M$, the high-level idea is to guess the labeling of agents adjacent to edges of~$M$, then perform local replacement of these edges with small gadgets that turn the social network into a prelabeled tree, and find an optimal extension of such a labeling using the algorithm from \Cref{lem:LMRMI:poly:tree}.

\begin{theorem}
\label{thm:MI:FPT:fes}
    \probName{$q$-Majority Illusion} is in \FPT when parameterized by the feedback edge set number $\operatorname{fes}(G)$ of the social network $G$.
\end{theorem}
\begin{proof}
Let $(\agents,G,q)$ be our input instance. We assume w.l.o.g.\ that $G$ is connected.
    Let $M$ be a feedback edge set of~$G$ of size $\fes(G)$.
    Let $P$ denote the set of agents that are adjacent to some edge in~$M$; then $|P| \leq 2|M|$.
    Let $\optcol$ be a labeling of~$G$ with blue as the majority color that maximizes the number $\umi[f](\agents)$ of agents under illusion in~$G$.
    
    First, for each $v \in P$, we guess its  label in~$\optcol$. 
    Second, we construct a new social network~$T$ that will be a tree by deleting all edges of~$M$ in~$G$, and then adding for each $v \in P$ exactly $|\{\{v,u\} \in M:\optcol(u)=\blue\}|$ blue and exactly $|\{\{v,u\} \in M:\optcol(u)=\red\}|$ red pendant leaves. 
    Then, to each of these $2|M|$ newly added agents, we add a pendant leaf having the opposite color. This way, we have added exactly $2|M|$ red and exactly $2|M|$ blue agents to~$N$; let $A$ denote the set of these $4|M|$ auxiliary agents. 
    
    We set $L(v)=\{\optcol(v)\}$ for each $v \in P$, 
    set $L(a)=\{c\}$ for each $a \in A$ having color~$c$, and set $L(v)=\{\red,\blue\}$ for all agents $v \in N \setminus P$. 
    Notice that for any labeling~$f'$ of~$G$ that coincides with~$f$ on~$P$, each agent $v \in P$ is under majority illusion in~$G$ if and only if it is under majority illusion in the corresponding labeling~$f'_T$ of~$T$ that coincides with $f'$ on~$N$ and satisfies $f'_T(a) \in L(a)$ for each $a \in A$. Moreover, such a labeling~$f'$ has blue as its majority color if and only if $f'_T$ has blue as the majority color.

    Next, we solve the \ListMI problem on~$(N \cup A, T,L)$ using the algorithm from \Cref{lem:LMRMI:poly:tree}. Let~$\mu_T$ be the returned solution value.
    
    Finally, we compute the answer to our instance as follows. Let $\mu_A$ denote the number of agents in~$A$ under majority illusion under our labeling of~$P \cup A$. By our previous arguments
    and the correctness of \Cref{lem:LMRMI:poly:tree}, the maximum number of agents under majority illusion under any labeling with majority color blue that coincides with~$f$ on~$P$ is exactly $\mu_T-\mu_A$.
    Hence, we return \Yes if $\mu_T-\mu_A \geq q|\agents|$ and \No otherwise.
    
    The  algorithm runs in $2^\Oh{\operatorname{fes(G)}}\cdot |\agents|^\Oh{1}$ time.
    
\end{proof}

\section{Bounded Distance to Disjoint Stars}
\label{sec:dds}

Let us now show that \probName{$q$-Majority Illusion} is intractable even if the social network~$G$ is close to being a collection of disjoint stars. 
To show the \Whness of \probName{$1$-Majority Illusion} when parameterized by $\dds(G)$, we provide a reduction from the following variant of \probName{Multidimensional Subset Sum}:  given an integer~$d$ and a set~$\calS$ of $n$ vectors in~$\mathbb{N}^d$, 
the task is to decide if there exists a subset~$\calS' \subseteq \calS$ with $\sum_{\myvec{s} \in \calS'}\myvec{s}=\sum_{\myvec{s} \in \calS \setminus \calS'} \myvec{s}$ and $|\calS'|=|\calS|/2$.
We call this the \probName{Multidimensional Halfset Sum} problem.

\begin{restatable}%
{lemma}{prophalfsetsum}
\label{lem:halfset-sum}
    \probName{Multidimensional Halfset Sum}  encoded in unary is \Wh with respect to the dimension.
\end{restatable}
\begin{proof}
    We give a reduction from \probName{Multidimensional Subset Sum}; let $I=(d,\calS,\myvec{t})$ of with $\calS=\{\myvec{s}_1,\dots,\myvec{s}_n\}$ be our input instance. We first append the value~$i$ to each vector~$\myvec{s}_i$ (increasing the dimension by one) and add $n$ new vectors~$\myvec{s}'_1,\dots,\myvec{s}'_n$ where $\myvec{s}'_i=(0,\dots,0,i) \in \mathbf{N}^{d+1}$ for each $i \in [n]$.
    We further add the two vectors $(\myvec{t},n^2)$ and $(\ssum-\myvec{t},n^2)$ where $\ssum=\sum_{\myvec{s} \in \calS}$.
    Note that the total sum of the vectors in the constructed instance $(d+1,\calS')$ of \probName{Multidimensional Halfset Sum} adds up to $(2\ssum,2n^2+2\binom{n}{2})$.
    It is straightforward to check that a solution~$X$ for~$I$ implies that there are $n+1$ vectors in $\calS' $ that sum up to $(\ssum,n^2+\binom{n}{2})$ can be extended to a solution in the constructed instance by adding those vectors $\myvec{s}'_i$ for which $\myvec{s}_i \notin X$ together with the vector $(\ssum-\myvec{t},n^2)$. Conversely, any solution to the constructed instance partitions the vectors in~$\calS'$ into two sets of size~$n+1$ with one of them, say~$X'$, containing the vector $(\ssum-\myvec{t},n^2)$. Then the remaining $n$ vectors in~$X'$ must add up to~$(\myvec{t},\binom{n}{2})$ which yields that $\{\myvec{s}_i:(\myvec{s}_i,i) \in X'\}$ is a solution for~$I$. 
\end{proof}

Now, we are ready to prove the main result of this section.

\begin{restatable}%
{theorem}{thmMIWhardfvs}
    \label{thm:MI:W1hard:dds}
    \probName{$q$-Majority Illusion} is \Wh  when parameterized by the distance to disjoint stars $\dds(G)$ of the social network~$G$.
\end{restatable}
\begin{proof}
    We give a reduction from \probName{Multidimensional Halfset Sum}, which is \Wh as shown in \Cref{lem:halfset-sum}, using similar ideas as in the reduction proving \Cref{thm:MI:W1hard:dcg}.
    
    Let $I=(d,\calS)$ be our input with $\calS=\{\myvec{s}_1,\dots,\myvec{s}_n\}$; we will use the notation $\ssum=\sum_{i \in [n]}\calS$.
    We may clearly assume that $n \geq 2$, and also that $d\geq 2$, as otherwise we can just append an additional coordinate with value~$1$ to each vector in~$\calS$.
    We may further assume that $\ssum[j] \geq 3$ for each $j \in [d]$, 
    as otherwise we can multiply all vectors by three; this ensures the required condition unless $\ssum[j]=0$ in which case we can  ignore the $j^{\text{th}}$ coordinates.
    
    We create an instance $J=(N,G,q)$ of \probName{$q$-Majority Illusion} as follows.

    \proofsubparagraph{Construction} 
    Let $\myvec{t}=\ssum/2$ denote our target vector. We define 
    $\beta=\max_{j \in [d]} \ssum[j]$
    and $\alpha=\frac{n}{2}(d\beta+1)-n-d-1$.
    
    First, for each $i \in [\alpha]$, we create a pair $(a_i,a'_i)$ of agents connected by an edge in~$G$; let $A=\{a_1,\dots,a_d, a'_1,\dots,a'_\alpha\}$. 
    Next, for each $i \in [n]$ we create a star with \emph{center agent}~$c_i$ and a set~$L_i$ of $2d\beta+1$ leaves with 
    \[
    L_i=L'_i \cup \bigcup_{j \in [d]} (L_i^j \cup \ol{L}_i^j) \cup \{\ell_i^0\} 
    \] 
    where $|L'_i|=d\beta$, and 
    $|L_i^j|=\myvec{s}_i[j]$ 
    and $|\ol{L}_i^j|=\beta-\myvec{s}_i[j]$ 
    for each $j\in [d]$.
    Let $C$ denote the set of all central agents.
    Additionally, we create a set~$M=\{m_1,\dots,m_d,\ol{m}_1,\dots,\ol{m}_d\}$ of \emph{modulator agents}. 
    For each $j \in [d]$ we connect agent~$m_j$ to all agents in~$\bigcup_{i \in [n]} L_i^j$ as well as to $a_j$, and we connect $\ol{m}_j$ to all agents in~$\bigcup_{i \in [n]} \ol{L}_i^j$ as well as to~$a'_j$.
    This finishes our definition of the social network~$G$ over agent set~$N$; note that the number of agents is $|N|=2\alpha+n (2d\beta+2)+2d$.
    We set the value of~$q$ such that $qN=|N|-n/2$. 

    Removing all modulator agents from~$G$ leaves a collection of disjoint stars; hence, $\dds(G) \leq |M|=2d$, so the presented reduction is a parameterized one. 
    We now show that $I$ is a yes-instance of \probName{Multidimensional Halfset Sum} if and only if $J$ is a yes-instance of \probName{$q$-Majority Illusion}.

    \proofsubparagraph{Correctness}
    Let us first assume that there exists a red---blue labeling~$\rblab$ for~$G$ with majority color blue where at least $q|N|$ agents are under majority illusion; let us choose $\rblab$ such that it maximizes the number of agents under majority illusion and, subject to that, maximizes the number or red agents in~$A$.
    We can assume that 
    \begin{equation}
        \label{eq:totalreds-fvs}
        \begin{split}
        \RedNo_\rblab(N)&=\frac{|N|}{2}-1=
        \alpha+n(d\beta+1)+d-1\\
        &=
        2\alpha+\frac{n}{2}(d\beta+1)+n+2d,
        \end{split}
    \end{equation}    
    as otherwise we can simply re-label some blue vertices as red as long as we keep blue the majority color.
  
    First note that each center vertex is red under~$\rblab$. Indeed, since all agents except for center and modulator agents have at most two neighbors in~$G$, by $\RedNo_\rblab(N)>n+d$ we know that if some center agent~$c_i$ is not red, then we can modify~$\rblab$ by re-labeling~$c_i$ as red and, if necessary, setting some agent in~$N \setminus (C \cup M)$ as blue; this way, the number of agents under majority illusion in~$L_i$ increases by at least $|L'_i \cup \{\ell_i^0\}|=d\beta+1 \geq 3$ while at most two agents may cease being under illusion; a contradiction to our choice of~$\rblab$.
    We can apply an analogous reasoning for modulator agents: since all center agents are red, re-labeling some modulator agent~$m_j$, $j\in [d]$ as red increases the number of agents under majority illusion by $\sum_{i \in [n]} |L_i^j|=\ssum[j]$; similarly, 
    re-labeling agent~$\ol{m}_j$ as red increases this number by $\sum_{i \in [n]}|\ol{L}_i^j| = n\beta-\ssum[j]$, both of which values are at least~$3$ by our initial assumptions; hence, the above arguments show that all modulator agents are red. 
    
    We next show that each agent in~$A$ is red.
    On the one hand, re-labeling a blue agent in~$A$ as red increases the number of agents under majority illusion in~$A$ by exactly one (recall that all modulator agents are red). On the other hand, some center agent~$c_i$  must be under majority illusion by ${q|N|>|N|-|C|}$; it follows that $c_i$ must have a red neighbor~$\ell$ in~$L_i$ not adjacent to any modulator agent. Hence, we can re-label~$\ell$ as blue (destroying the illusion only for~$c_i$) to compensate for re-labeling some agent in~$A$ as red, without decreasing the number of agents under majority illusion. This contradicts our choice  of~$\rblab$, proving that all agents in~$A$ are red.

    Notice now that since all agents in~$M \cup C \cup A$ are red, from~\eqref{eq:totalreds-fvs} we get 
    \begin{equation}
    \label{eq:red-leaves-fvs}
        \RedNo_\rblab(L)=\frac{n}{2}(d\beta+1). 
    \end{equation}
    Since each center agent needs at least $d\beta+1$ red agents in~$L_i$ to be under majority illusion, and at least $\frac{n}{2}$ center agents must be under majority illusion by $q|N|=|N|-\frac{n}{2}$, it follows that there must be exactly $\frac{n}{2}$ center agents under majority illusion, each of them having exactly $d\beta+1$ red neighbors. Moreover, using again $q|N|=|N|-\frac{n}{2}$, we obtain that all modulator agents must be under majority illusion as well. 

    Define $\calS'=\{\myvec{s}_i: i \in [n], c_i$ is under majority illusion$\}$; then we know $|\calS'|=|\calS|/2$. 
    Since each modulator agent~$m_j$ for some~$j \in [d]$ has one red neighbor in~$A$, it    
    needs at least $\myvec{t}[j]$ red neighbors from~$\bigcup_{i \in [n]}L_i^j$
    as otherwise it would not be under majority illusion by 
    \[\RedNo_\rblab(N_G(m_j)) \leq \myvec{t}[j] =  
        \ssum[j]-\myvec{t}[j] < \BlueNo_\rblab(N_G(m_j)).\]
    This yields 
    \begin{equation}
    \label{eq:halfsum-atleast-target}
        \sum_{\myvec{s}\in \calS'} \myvec{s}[j] \geq \myvec{t}[j] \qquad \text{for each $j \in [d]$}.
    \end{equation}
    
    Similarly, since each modulator agent~$\ol{m}_j$ for some~$j \in [d]$ has one red neighbor in~$A$, it    
    needs at least $\frac{n}{2}\beta-\myvec{t}[j]$ red neighbors from~$\bigcup_{i \in [n]}\ol{L}_i^j$
    as otherwise it would not be under majority illusion by 
    \begin{align*}
    \RedNo_\rblab(N_G(\ol{m}_j)) &< 
    \frac{n}{2}\beta-\myvec{t}[j] \\
    & =  
    \frac{n}{2}\beta -(\ssum[j]-\myvec{t}[j]) \leq  \BlueNo_\rblab(N_G(\ol{m}_j)).
    \end{align*}
    This yields 
     \[   \sum_{\myvec{s}\in \calS'} \myvec{s}[j] \leq \myvec{t}[j] \qquad \text{for each } j \in [d].
    \]
    which, together with~\eqref{eq:halfsum-atleast-target} yields that $\sum_{\myvec{s}\in \calS'} \myvec{s} = \myvec{t}$, proving that $\calS'$ is a solution for~$I$.
    
    \medskip
    For the other direction, assume now that $I$ admits a solution~$\calS'$.
    We create a red--blue labeling~$\rblab$ for~$G$ as follows. We create all agents in~$C \cup M \cup A$ red, as well as the agents in $L_i \setminus L'_i$ for each $i \in [n]$ where $\myvec{s}_i \in \calS'$. 
    The number of red agents thus satisfies \eqref{eq:totalreds-fvs}, so the majority color in~$\rblab$ is blue.
    It is clear that all agents in~$A \cup L \cup \{c_i:\myvec{s}_i \in \calS'\}$ are under majority illusion.
    We show that all modular agents are under majority illusion as well. First, each agent~$m_j$ for some $j \in [d]$ has 
    exactly $1+\sum_{\myvec{s} \in \calS'} \myvec{s}[j]=1+\myvec{t}[j]$ red neighbors and exactly $\sum_{\myvec{s} \in \calS \setminus \calS'} \myvec{s}[j]=\myvec{t}[j]$ blue neighbors.
    Second, each agent~$\ol{m}_j$ for some $j \in [d]$ has 
    exactly $1+\frac{n}{2}\beta - \sum_{\myvec{s} \in \calS'} \myvec{s}[j]=1+\frac{n}{2}\beta-\myvec{t}[j]$ red neighbors and exactly $\frac{n}{2}\beta-\sum_{\myvec{s} \in \calS \setminus \calS'} \myvec{s}[j]=\frac{n}{2}\beta-\myvec{t}[j]$ blue neighbors.
    Thus, all modulator agents are under majority illusion, implying that all agents except for the $\frac{n}{2}$ center agents in~$\{c_i:\myvec{s}_i \in \calS \setminus \calS'\}$ are under majority illusion. This proves that $J$ is a yes-instance of \probName{$q$-Majority Illusion}. 
\end{proof}

The strength of this intractability result is demonstrated by its consequences on other prominent graph parameters:

\begin{corollary}
\label{cor:MI:dds:strong}
    \probName{$q$-Majority Illusion} is \Wh  when parameterized by the combination of the following parameters%
        of the social network~$G$: 
    \begin{itemize}
        \item the feedback vertex set number $\fvs(G)$,
        \item the $4$-path vertex cover number $\operatorname{4pvc}(G)$, %
        and
        \item the treedepth~$\td(G)$.
    \end{itemize} 
\end{corollary}

A reduction from the \probName{Capacitated Vertex Cover} problem, proved by~\citeauthor{DomLSV2008}~\shortcite{DomLSV2008} to be \Wh when parameterized by treedepth, shows that even the case $q=1$ remains intractable when parameterized by treedepth:

\begin{restatable}%
{theorem}{thmMIWhardtw}
\label{thm:MI:W1hard:td}
    \probName{$1$-Majority Illusion} is \Wh  parameterized by the treedepth $\td(G)$ of the social network~$G$.
\end{restatable}
\begin{proof}
    We present a parameterized reduction from \probName{Capacitated Vertex Cover} which is \Wh when parameterized by the treedepth of the input graph~\cite{DomLSV2008}. Let $(H,c,k)$ be our input instance;
    we assume w.l.o.g.\ that $|V| \geq 2k+2$. 
    We are going to construct an instance $(N,G,1)$ of \probName{$q$-Majority Illusion}.

    \proofsubparagraph{Construction}
    We are going to create a vertex gadget for each $v \in V$, an edge gadget for each $e \in E$, and a single counting gadget. Our social network will consist of the disjoint union of these gadgets, with some additional inter-gadget edges running between them.
    
    We start by introducing a general structure called an \emph{$(v,\ell _r,\ell_b)$-cap} that will be used in all vertex gadgets as well as in the counting gadget. Here, $v$ is an agent, and $\ell_r$ and~$\ell_b$ are non-negative integers. %
    An $(v,\ell_r,\ell_b)$-cap $Q$ is defined over an agent set $\{v\} \cup A_Q$  where 
    \begin{align*}
    A_Q &= \bigg(\bigcup_{i \in [\ell_r]} A_{Q,i}\bigg) \cup \bigg( \bigcup_{i \in [\ell_b]} \cup B_{Q,i}\bigg) \\
    A_{Q,i}& =\{a_{i,j},a'_{i,j}: j \in [4]\} \\
    B_{Q,i}& =\{b_{i,j},b'_{i,j}: j \in [4]\} 
    \end{align*}

    \begin{figure}
        \centering
        \includegraphics[width=\columnwidth]{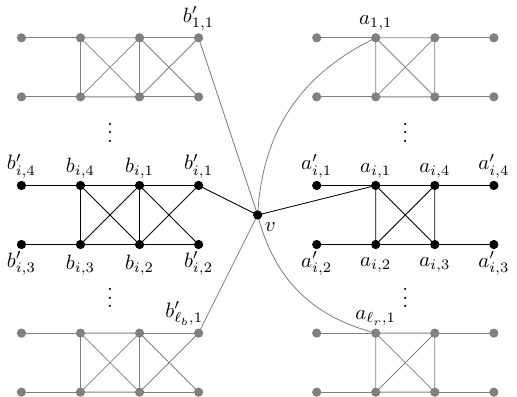}
        \caption{Illustration of a $(v,\ell_r,\ell_b)$-cap.}
        \label{fig:cap}
    \end{figure}

    The social networks formed by the agents in~$A_{Q,i}$ are the same for each $i \in [\ell_r]$, and similarly for the social networks formed by the agents in~$B_{Q,i}$ for all $i \in [\ell_b]$; see \Cref{fig:cap} for an illustration. 
    Formally,  for each $i \in [\ell_r]$, agents $a_{i,j}$ for $j \in [4]$ form a clique, and $a'_{i,j}$ is connected to~$a_{i,j}$ for each $j \in [4]$. 
    Similarly, for each $i \in [\ell_b]$, agents~$b_{i,j}$ for $j \in [4]$ form a clique, and $b'_{i,j}$ is connected to~$b_{i,j}$ for each $j \in [4]$. 
    We further connect agent~$b'_{i,1}$ to~$b_{i,2}$, and agent~$b'_{i,2}$ to~$b_{i,1}$ for each $i \in [\ell_r]$. 
    Finally, we connect agent~$a$ to~$a_{i,1}$ for each $i \in [\ell_r]$ and
    and to~$b'_{i,1}$ for each $i \in [\ell_b]$. 
    The key property of a $(v,\ell,\ell')$-cap is stated in the next claim.
    \begin{claim}
    \label{clm:cap-prop}
        Let $Q$ be an $(v,\ell_r,\ell_b)$-cap with a red--blue labeling~$\rblab$. 
        If all agents in~$Q$ other than~$v$ have more red than blue neighbors, then 
        \begin{align} 
        \label{eq:clm-cap-c1}
        \RedNo_\rblab(A(Q) \setminus \{v\})
        & \geq 
        \BlueNo_\rblab(A(Q) \setminus \{v\})
        \\
        \label{eq:clm-cap-c2}
        \RedNo_\rblab(N_Q(v))
        &\geq \ell_r; \\
        \label{eq:clm-cap-c3}
        \BlueNo_\rblab(N_Q (v))
        &\leq \ell_b. 
        \end{align}
        Equality for~\eqref{eq:clm-cap-c1} implies equality in~\eqref{eq:clm-cap-c2} and~\eqref{eq:clm-cap-c3}.
        Moreover, $B$ admits  
        red--blue labelings~$\rblab$ and~$\rblab'$ with $\rblab(v)=\blue$ but $\rblab'(v)=\red$
        that both satisfy~\eqref{eq:clm-cap-c1}--\eqref{eq:clm-cap-c3} with equalities.        
    \end{claim}
    \begin{claimproof}
        Consider some index~$i$. Since each agent $a'_{i,j}$ for $j \in \{2,3,4\}$ and each agent~$b'_i$ for $j \in [4]$ has at most two neighbors in~$Q$, those neighbors must be red. Thus, all agents in~$\{a_{i,j},b_{i,j}:j \in [4]\}$ are red under~$\rblab$. As the number of such agents is $4(\ell_r+\ell_b)$, and the total number of agents in~$Q$ not counting~$v$ is $8(\ell_r+\ell_b)$,  statements~\eqref{eq:clm-cap-c1}--\eqref{eq:clm-cap-c3} follow. Assuming that the number of red and blue agents is equal, we get that all remaining agents must be blue. In particular, 
        $v$ has $\ell_r$ red neighbors (namely, $a_{1,1},\dots,a_{\ell_r,1}$) and $\ell_b$ blue ones (namely, $b'_{1,1},\dots,b'_{\ell_b,1}$).
        Finally, observe that such a labeling indeed ensures that all agents have more red than blue neighbors, irrespective of the labeling of~$v$, implying the last statement of the claim.
    \end{claimproof}
    
    Now, the \emph{vertex gadget} $G_v$ for some vertex $v \in V$, which is 
    simply a $(v,c(v)+2,\deg_H(v)-c(v))$-cap over a set $A_v$ of agents where $\deg_H(v)$ is the degree of~$v$ in~$H$.

    The \emph{edge gadget}~$G_e$ for some edge $e=\{u,v\} \in E$ is defined over agent set~$A_e=Z_e \cup W_e$
    where 
    \begin{align*}
        Z_e &= \{z_i^{x,e},z_i^{v,e}:i \in [4],x \in \{u,v\}\} \qquad \textrm{ and} \\
        W_e &= \{w^{x,e}_4,w^{x,e}_{2,3},w^{x,e}_{3,4}:x \in \{u,v\}\}.
    \end{align*} 

    \begin{figure}
        \centering
        \includegraphics[width=\columnwidth]{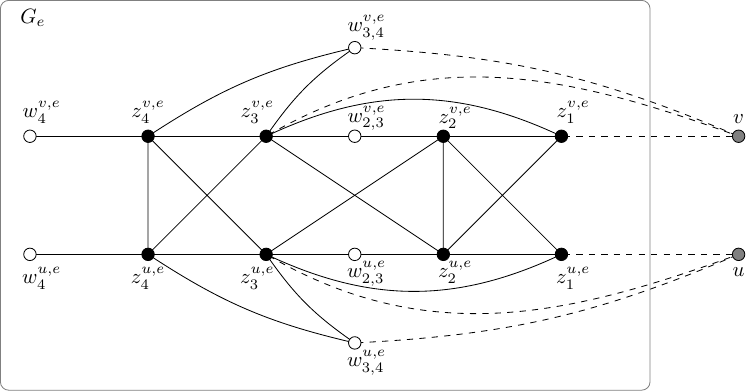}
        \caption{Illustration of an edge-gadget $G_e$ for $e=\{u,v\}$. Inter-gadget connections are shown with dashed lines. Agents in~$Z_e$, $W_e$, and those outside~$G_e$ are shown as black, white, and gray circles, respectively.}
        \label{fig:edge-gadget}
    \end{figure}
    
    The social gadget underlying~$G_e$ is depicted in \Cref{fig:edge-gadget}; the construction is symmetric for the two endpoint of~$e$.
    Formally, we define the connections within~$G_e$ by connecting~$z_2^{v,e}$ to agents $z_1^{v,e}$, $z_1^{u,e}$, $w^{v,e}_{2,3}$, $z_2^{u,e}$, and $z_3^{u,e}$, 
    connecting~$z_3^{v,e}$ to agents $z_1^{v,e}$, $z_4^{v,e}$, $w^{v,e}_{2,3}$, $w^{v,e}_{3,4}$, $z_2^{u,e}$, and $z_4^{u,e}$, 
    and additionally connecting~$z_4^{v,e}$ to agents~$w^{v,e}_{4}$, $w^{v,e}_{3,4}$ and~$z_4^{u,e}$.
    Finally, we make the construction symmetric by adding the analogous edges after switching the roles of~$u$ and~$v$.

    Next, we define the \emph{counting gadget} $G_{a^\star}$ that is based on an $(a^\star,|V|-k+1,k)$-cap for a newly introduced agent~$a^\star$ over agent set~$Q_{a^\star}$, 
    and additionally contains a set~$K_{a^\star}$ of $|V|-2k$ agents whose underlying social network is a collection of cliques of size~$2$ and~$3$.

    We finish the construction by defining all inter-gadget edges as follows: for each vertex~$v \in V$, we connect agent~$v$ with agents~$z_1^{v,e}$ ,$z_3^{v,e}$, and~$w^{v,e}_{3,4}$ for all edges~$e \in E$ incident to~$v$ in~$H$.
    Furthermore, we connect agent~$a^\star$ in the counting gadget to each agent~$v$ for $v \in V$.

    Note that $G$ can be computed in polynomial time. The next claim shows that the presented reduction is in fact a parameterized one in the sense that $\td(G)$ is a function of~$\td(H)$.

    \begin{claim}
        The constructed social network~$G$ has treedepth at most $9\td^2(H))$.
    \end{claim}
    \begin{claimproof}
        We show how to construct a treedepth decomposition for~$G$ of the claimed depth, using a treedepth decomposition~$\calT_H$ 
        of depth $\td(H)$ 
        for~$H$. 
        
        First, it is straightforward to check that for each~$v \in V$, deleting agent~$v$ from the vertex gadget~$G_v$ leaves us with a social network with treedepth at most~$5$. Hence, $G_v$ admits a treedepth decomposition of depth at most~$6$ with $v$ as its root.
        Similarly, the counting gadget also admits a treedepth decomposition~$\calT_c$ of depth~$6$ with~$a^\star$ as its root.
        It is also easy to verify that for each edge~$e \in E$ connecting vertices~$u$ and~$v$ in~$H$, the social network induced by the edge-gadget~$G_e$ and the two agents~$u$ and~$v$ admits a treedepth decomposition~$\calT_{e}$ of depth~$10$ whose root is~$u$ with~$v$ as its only child. 
        
        We now show how to compute a treedepth decomposition for~$G$ based on~$\calT_H$.
        We start by setting~$\calT:=\calT_H$, and proceed as follows.
        \begin{itemize}
            \item We add $\calT_c$ to~$\calT$ by setting $a^\star$ as the new root and adding the root of~$\calT$ as a child of~$a^\star$ (besides all its children in~$\calT_c$).
            \item For each $v \in V$, we add $\calT_v$ to $\calT$ by identifying the root of $\calT_v$ with the node~$v$ in~$\calT$.
            \item We iteratively apply the following operation for each  $e \in E$. Let $u$ and~$v$ be the endpoints of~$e$ such that $v$ is a descendant of~$u$. We delete~$u$ from~$\calT_e$ so that $v$ is the root of the obtained tree $\calT_e-u$, and then add~$\calT_e-u$ by identifying the vertex~$v$ within~$\calT$ with the root of~$\calT_e-u$, and re-position each child of~$v$ within~$\calT$ as the child of some (arbitrarily chosen) leaf of~$\calT_e-u$. 
            Note that $u$ becomes an ascendant of all vertices of~$G_e$ within~$\calT$, and moreover, all vertices that were descendants (ascendents) of~$v$ before the operation remain descendants (ascendents, respectively) of~$v$ afterwards.
        \end{itemize}
        
        It is easy to see that the above procedure yields a treedepth decomposition for~$H$.
        To see that it has the claimed depth of~$9\td^2(H)$, first observe that the operations involving the counting gadget and the vertex gadgets increase the depth of~$\calT$ by at most~$5$. 
        Second, note that the number of times when the agents within some edge-gadget need to be added right under some vertex~$v$ is not more than the number of its ascendants in~$H$, i.e., at most~$\td(H)$. Since each such operation adds $8$ to the depth of~$\calT$, the incurred total increase in depth is at most~$8\td(H)$ for these additions at~$v$. Summing such increases up for each vertex and recalling the depth~$\td(H)$ of~$\calT_H$, we obtain that the obtained treedepth decomposition indeed has depth not more than~$9\td^2(H)$.
    \end{claimproof}

    \proofsubparagraph{Correctness}
    It remains to show that $(H,c,k)$ is a yes-instance of \probName{Capacitated Vertex Cover} if and only if $(N,G,1)$ is a yes-instance of \probName{$q$-Majority Illusion}.
    
    \smallskip
    Suppose first that $G$ admits a red--blue labeling~$\rblab$ which ensures that all agents are under majority illusion.

    We use the following claim about 
    edge gadgets.
    \begin{claim}
    \label{clm:edge-gadget}
        Let~$\{u,v\}=e \in E$. If all agents in~$G_e$ have more red than blue neighbors under~$\rblab$, then 
        \begin{equation}
        \label{eq:edgegadget-balanced}
            \RedNo_\rblab(A_e)
            \geq 
            \BlueNo_\rblab(A_e)
        \end{equation}
     and equality may only happen if 
     \begin{enumerate}[label=(\roman*),left=2pt]
         \item  exactly one of~$z_1^{v,e}$ and~$z_1^{u,e}$ is blue and
         \item $\rblab(z_1^{v,e})=\blue$ implies $\rblab(v)=\red$
         whereas 
         $\rblab(z_1^{u,e})=\blue$ implies $\rblab(u)=\red$;
         \item both~$u$ and~$v$ have one red and one blue neighbor within~$G_e$ that is not in~$\{z_1^{v,e},z_1^{u,e}\}$.
     \end{enumerate}  
     Conversely, if we fix at least one of~$u$ and~$v$ as red, then agents in~$G_e$ can be labeled such that \eqref{eq:edgegadget-balanced} and (i)--(iii) hold, and each agent in~$G_e$ has more red than blue neighbors.
    \end{claim}

    \begin{claimproof}
        First note that 
        all four agents in~$W_e \setminus \{w^{u,e}_{3,4},w^{v,e}_{3,4}\}$ have at most two neighbors in~$G$; hence,
        all their neighbors must be red. 
        This means that agents $Z_e \setminus \{z_1^{u,e},z_1^{v,e}\}$
        are red. It is not possible that all remaining eight agents of~$G_e$ are blue, as that would yield more blue than red neighbors for $z_2^{v,e}$ and~$z_2^{u,e}$. Hence, there are at least $7$ red agents, implying~\eqref{eq:edgegadget-balanced}. 
        
        If \eqref{eq:edgegadget-balanced} holds with equality, then there must be exactly one additional red agent in~$G_e$, which means that exactly one common neighbor of $z_2^{v,e}$ and $z_2^{u,e}$ must be red: either~$z_1^{v,3}$ or $z_1^{u,e}$. Hence, (i) and (iii) hold. 
        If $z_1^{v,3}$ is red, then $z_1^{u,e}$ is blue, and thus inspecting the neighborhood of~$z_3^{u,e}$ we obtain that $u$ must be red. Similarly, if $z_1^{v,e}$ is blue, then $v$ must be red. This proves (ii).

        Finally, assume w.l.o.g.\ that we label~$v$ red. Then it is easy to verify that setting all six agents in~$W_e$ as well as $z_1^{v,e}$ blue and all other agents red yields a labeling with the required properties.
    \end{claimproof}
    
    Let $S=\{v \in V:\rblab(v)=\red\}$;  
    we are going to prove that $S$ is a capacitated vertex cover of~$H$. 
    First, we show the following.

    \begin{claim}
        \label{clm:vc-size}
        We have 
        $|S|=\RedNo_\rblab(V)\leq k$.
        Moreover,
        \begin{align}
        \label{eq:vertexbalance}
            \RedNo_\rblab(A_v \setminus \{v\})=
            \BlueNo_\rblab(A_v \setminus \{v\})
        \end{align}
        for each $v \in V$,
        and        
        \eqref{eq:edgegadget-balanced} holds with equality for each $e \in E$. 
    \end{claim}
    \begin{claimproof}
    First note that all agents in~$K_{a^\star}$ must be red, as each of them has at most two neighbors, and thus all these neighbors must be red. 
    Second, consider the set of neighbors of~$a^\star$ within the counting gadget~$G_{a^\star}$; note that these are all agents in $Q_{a^\star}$. 
    Applying the first statement of \Cref{clm:cap-prop} for the $(a^\star, |V|-k+1,k)$-cap over~$Q_{a^\star}$, we obtain that 
    $a^\star$ has at least $|V|-k+1$ red and at most $k$ blue neighbors within~$Q_{a^\star}$. 
    Hence, we can define a non-negative $\Delta \in \mathbb{N}$ for which 
    \begin{align}
    \label{eq:astar-neighbors}
        \begin{split}
        \RedNo_\rblab(N_G(a^\star)) \cap Q_{a^\star}
        & =|V|-k+1+\Delta \quad \textrm{ and} 
        \\
        \BlueNo_\rblab(N_G(a^\star)) \cap Q_{a^\star} &= k-\Delta.
        \end{split}
    \end{align}
    By the arguments in the proof of \Cref{clm:cap-prop}, we also get 
    \begin{align}
    \label{eq:astar-total}
        \begin{split}
        \RedNo_\rblab(Q_{a^\star}\setminus \{a^\star\}) 
        & \geq |Q_{a^\star}|-1)/2+\Delta \quad \textrm{and}\\
        \BlueNo_\rblab(Q_{a^\star} \setminus \{a^\star\}) 
        & \leq (|Q_{a^\star}|-1)/2-\Delta
        \end{split}
    \end{align}

    Recall now that the set of neighbors of~$a^\star$ outside the counting gadget is~$V$, and the red agents among these are those in~$S$. Thus, from~\eqref{eq:astar-neighbors} we obtain 
    \begin{align}
    \label{eq:astar-allneighbors}
        \begin{split}
        \RedNo_\rblab(N_G(a^\star))
        & =|V|-k+1+\Delta+|S| \qquad \textrm{ and} 
        \\
        \BlueNo_\rblab(N_G(a^\star)) &= k-\Delta+|V|-|S|.
        \end{split}
    \end{align}
    Since $a^\star$ is under majority illusion, we get
    \begin{align}
        \notag 
        \RedNo_\rblab(N_G(a^\star)) & > 
        \BlueNo_\rblab(N_G(a^\star)) \\\notag 
        |V|-k+1+\Delta+|S| & >k-\Delta+|V|-|S| \\\notag 
        2|S|+2\Delta +1 & > 2k \\
        |S|+|\Delta| &\geq k. \label{eq:SplusDelta}
    \end{align}

    Let $N'$ denote the set of agents contained in some vertex- or edge-gadget. 
    \Cref{clm:cap-prop} implies that 
    \begin{equation*}
    \RedNo_\rblab(A_v \setminus \{v\})  
    \geq \BlueNo_\rblab(A_v \setminus \{v\})  
    \end{equation*}
    holds for each $v \in V$.
    Summing up the above inequalities for each $v \in V$ and \eqref{eq:edgegadget-balanced} for each $e \in E$, we get that $N'$ contains at least as many red as blue agents, i.e., 
    \begin{equation}
    \label{eq:Nprime-balanced}
        \RedNo_\rblab(N') \geq 
        \BlueNo_\rblab(N').
    \end{equation}
    
    Counting now the number of red agents in~$K_{a^\star}$, $Q_{a^\star}$, $N'$, and~$V$ and using \eqref{eq:astar-total} and \eqref{eq:Nprime-balanced}, we get that 
    \begin{equation} 
    \label{eq:redtotal}
     \RedNo_\rblab(N) \geq  (|V|-2k)+\frac{|Q_{a^\star}|-1}{2}+\Delta+\frac{|N'|}{2}+ |S|.
    \end{equation}
    Similarly counting blue agents in~$Q_{a^\star}$, $N'$, and $V$, we obtain 
    \begin{equation} 
    \label{eq:bluetotal}
    \BlueNo_\rblab(N) \leq \frac{|Q_{a^\star}|-1}{2}-\Delta+1 + \frac{|N'|}{2}+(|V|-|S|).
    \end{equation}
    Since blue is the majority opinion, i.e., $\RedNo_\rblab(N) < \BlueNo_\rblab(N)$, it follows that 
    \begin{align*}
        -2k+\Delta+|S| & < -\Delta+1-|S| \\
        2|S|+2\Delta &< 2k+1 \\
        |S|+\Delta &\leq k
    \end{align*}
    Taking into account \eqref{eq:SplusDelta}, we obtain that $|S|+\Delta=k$ must hold. Hence, the first statement of the claim follows.
    Observe that $|S|+\Delta=k$ implies also that \eqref{eq:Nprime-balanced} must hold with equality, as otherwise $\RedNo_\rblab(N') \geq |N'|/2+1$, 
    which would imply
    \begin{align*}
    \RedNo_\rblab(N) & \geq  (|V|-2k)+\frac{|Q_{a^\star}|-1}{2}+\Delta+\frac{|N'|}{2}+1+ |S| \\
     & =  (|V|-k)+\frac{|Q_{a^\star}|-1}{2}+\frac{|N'|}{2}+1 \\
    & \geq \BlueNo_\rblab(N)
    \end{align*}
    contradicting our assumption that blue is the majority winner.
    In particular, this means that \eqref{eq:vertexbalance} for each $v \in V$ and~\eqref{eq:edgegadget-balanced} for each $e \in E$ hold with equality.
    \end{claimproof}
    
    \Cref{clm:vc-size} immediately yields that (i) and~(ii) in \Cref{clm:edge-gadget} holds for each $e \in E$, and thus $S$ is a vertex cover for~$H$.  We next show that it is a capacitated vertex cover for~$H$.

    Let us define a mapping~$\mapvc:E \to S$ such that $\mapvc$ maps $e \in E$ to its endpoint~$v$ if and only if $\rblab(z_1^{v,e})=\blue$.
    By \Cref{clm:vc-size} we know that  (i) in \Cref{clm:edge-gadget} holds for each $e \in E$, and thus $\mapvc$ is well-defined. To show that $S$ is a capacitated vertex cover for~$H$ with capacity function~$c$, it remains to show that $f$ maps at most~$c(v)$ vertices to each $v \in S$.

    To see this, consider some~$v \in S$. 
    By \eqref{eq:vertexbalance} in \Cref{clm:vc-size} and using \Cref{clm:cap-prop} we know that $v$ has exactly $c(v)+2$ red neighbors and $\deg_H(v)-c(v)$ blue ones within~$G_v$. By the definition of~$\mapvc$, we further know that $v$ has exactly $|\mapvc^{-1}(v)|$ blue neighbors within edge-gadgets, and thus exactly $\deg_H(v)-|\mapvc^{-1}(v)|$ red ones. Finally, taking into account $a^\star$, we obtain that
    \begin{align}
    \label{eq:total-vertex-neighbors}
        c(&v)+ 2  + \deg_H(v)-|\mapvc^{-1}(v)| =
        \RedNo_\rblab(N_G(v)) \\
         &> 
        \BlueNo_\rblab(N_G(v)) =
        \deg_H(v)-c(v)+|\mapvc^{-1}(v)| +1
        \notag
    \end{align}
    which implies
    \begin{align*}
        2c(v)+1 & > 2|\mapvc^{-1}(v)|\\
        c(v)  &\geq |\mapvc^{-1}(v)|
    \end{align*}
    as required. Hence, $S$ is a capacitated vertex cover for~$H$ of size at most~$k$, and thus $(H,c,k)$ is a yes-instance.

    \medskip
    For the other direction, suppose that $S$ is a capacitated vertex cover for~$H$ of size~$k$ with mapping $\mapvc\colon E \to S$. We define a red--blue labeling~$\rblab$ as follows.
    Using \Cref{clm:cap-prop}, we label the agents in each vertex gadget~$G_v$, $v \in V$, so that 
    \begin{itemize}
        \item $\rblab(v)=\red$ if and only if $v \in S$,
        \item $v$ has $c(v)+2$ red and $\deg_H(v)-c(v)$ blue neighbors within the gadget~$G_v$,
        \item \eqref{eq:vertexbalance} holds for $v \in V$, and 
        \item each agent in~$G_v$ other than~$v$ has more red than blue neighbors.    
    \end{itemize} 
    Using \Cref{clm:edge-gadget}, we then label the agents in~$G_e$ for each $e \in E$ so that 
    \begin{itemize}
        \item for each endpoint~$v$ of $e$, we have $\rblab(z_1^{v})=\blue$ if and only if $\mapvc(e)=v$,
        \item properties (i)--(iii) in \Cref{clm:edge-gadget} hold,
        \item \eqref{eq:edgegadget-balanced} holds, and
        \item each agent in~$G_e$ other than~$v$ has more red than blue neighbors.
    \end{itemize}
    Finally, using \Cref{clm:cap-prop} again, we label the agents in the counting gadget~$G_{a^\star}$ so that 
    \begin{itemize}
        \item $\rblab(a)=\red$ for each agent $a \in K_{a^\star}$,
        \item $\rblab(a^\star)=\blue$,
        \item $a^\star$ has $|V|-k+1$ red and $k$ blue neighbors within~$G_{a^\star}$,
        \item the number of red and blue agents in~$Q_{a^\star} \setminus \{a^\star\}$ is equal, and 
        \item each agent in~$G_{a^\star}$ other than~$a^\star$ has more red than blue neighbors.    
    \end{itemize} 
    
    Note further that for each $v \in V$,
    by $|\mapvc^{-1}(v)|\leq c(v)$ we get that \eqref{eq:total-vertex-neighbors} holds. 
    Moreover, \eqref{eq:astar-allneighbors} holds with $\Delta=0$, and hence, $|S|=k$ yields 
    \[\RedNo_\rblab(N_G(a^\star))=|V|+1>|V|=\BlueNo_\rblab(N_G(a^\star)).
    \]
    This proves that each agent has more red than blue neighbors in~$G$ under~$\rblab$.
    
    In order to show that all agents are under majority illusion, it now suffices to show that blue is the majority winner.
    To see this, note that \eqref{eq:redtotal} and \eqref{eq:bluetotal} both hold with $\Delta=0$ and $|S|=k$, implying $\RedNo_\rblab(N)< \BlueNo_\rblab(N)$.
    Therefore, $(N,G,1)$ is a yes-instance of \probName{$q$-Majority Illusion}.

    This proves the correctness of our reduction.
\end{proof}

\section{Bounded Distance to Disjoint Paths}
\label{sec:ddp}

Note that the \XP algorithm for treewidth by~\citet{GrandiKLRT2025}, as well as our simpler (and faster) dynamic programming for \ListMI, presented in \Cref{lem:LMRMI:poly:tree} and running in $\Oh{|\agents|^5}$ time, implies that \probName{$q$-Majority Illusion} can be solved in polynomial time on trees. 
If the social network is a path~$P$, then the problem gets even easier: %
simply labeling every other agent on~$P$ as red maximizes the number of agents under majority illusion.  

\begin{theorem}
\label{thm:MI:path:linear}
    If the social network~$G$ is a path, \qMI can be solved in linear time.
    In fact, the maximum number of agents under majority illusion in~$G$, denoted by~$\umi[](G)$, depends only on the number~$|\agents|$ of agents:
    \begin{equation*}
        \umi[](G) = \begin{cases}
            0 & \text{if } |\agents| \leq 2, \\
            \left\lceil 
            |\agents|/2
            \right\rceil & \text{if } |\agents| \text{ is odd,}\\
            |\agents|/2 
            - 1 & \text{if } |\agents| \text{ is even.}
        \end{cases}
    \end{equation*}
\end{theorem}
\begin{proof}
    If $|\agents|\leq 2$, then all agents have to be labeled with the majority color. Thus, no agent is under the illusion. For $|\agents|\geq 3$, 
    let $f^*$ be a labeling of~$G$ for which $\umi[G]=\umi[f^\star](G)$.

    For a path~$P$, let $\mrn[\opt|x](P)$ denote the maximum number of agents with a majority-red neighborhood under any labeling~$f$ of~$P$ that satisfies $\RedNo_f(P)=x$. 
    The following claims establish the value of~$\mrn[\opt|x](P)$ depending on the parity of~$|P|$.

        For a path~$P=(a_1,\dots,a_\ell)$, we define the agent sets $P_1=\{a_i:i \in [\ell], i \text{ odd}\}$ and 
            $P_2=\{a_i:i \in [\ell], i  \text{ even}\}$.
        
        \begin{claim}
            \label{clm:MI:paths:even}
            For a path~$P=(a_1,\dots,a_\ell)$ with $\ell$ even, we have $\mrn[\opt|x](P)=x$ for each $x \in [\ell]$.    
        \end{claim}
        \begin{claimproof}
            Let $A_i$ be the agents with a majority-red neighborhood in~$P_i$ for each $i \in [2]$. Note that all neighbors of any agent in~$A_1 \cup A_2$ must be red. At most one agent (an endpoint of~$P$) has only one neighbor in each~$A_i$, $i \in [2]$, with every other agent having two neighbors, all of them red and in~$P_{3-i}$. Thus, summing up the neighborhood sizes of all agents in~$A_i$ yields at least $2|A_i|-1$. As each red agent in~$P_{3-i}$ can contribute to the neighborhood of at most two agents in~$A_i$, there must be at least $|A_i|$ red agents in~$P_{3-i}$. Hence, labeling $x$ agents red in~$P$ yields at most~$x$ agents with a majority-red neighborhood. Moreover, such a labeling exists for each $x \in [\ell]$: it can be verified easily that it suffices to label the first~$x$ the agents according to the order $a_2,a_4,\dots,a_{\ell}, a_{\ell-1},a_{\ell-3},\dots,a_1$ as red.   
         \end{claimproof}
            
        Now, we show the analogous result for the case when the length of $P$ is odd.
        
        \begin{claim}
            \label{clm:MI:paths:odd}
            For a path~$P=(a_1,\dots,a_\ell)$ with $\ell$ odd, we have 
            \begin{equation}
                \label{eq:MI:path:opt-for-odd}   
                \mrn[\opt|x](P)=
                \begin{cases}
                    x+1 &  \text{ if } x = \frac{\ell-1}{2}; \\
                    x &  \text{ if } x  \in [\ell], x \neq \frac{\ell-1}{2}.
                \end{cases}
            \end{equation}
        \end{claim}
        \begin{claimproof}
            Let $A_i$ be the agents with a majority-red neighborhood in~$P_i$ for each $i \in [2]$. Again, all neighbors of any agent in~$A_1 \cup A_2$ must be red. Note that at most two agents in~$A_1$ (the endpoints of~$P$) have only one neighbor, with every other agent in~$A_1 \cup A_2$ having two neighbors, all of them red. 
            
            Consider first $A_1$. Summing up the neighborhood sizes of all agents in~$A_1$ yields at least $2|A_1|-2$. Since  each red agent in~$P_2$ can contribute to the neighborhood of at most two agents in~$A_1$, there must be at least $|A_1|-1$ red agents in~$P_2$; moreover, equality is only possible if $A_1=P_1$ (as otherwise there is at least one red agent in~$P_2$ with only one neighbor in~$A_1$).   Hence, labeling $x$ agents red in~$P_2$ yields at most~$x$  agents with a majority-red neighborhood, unless $x=|P_2|$ in which case all $x+1$ agents in~$A_1$ have majority-red neighborhood. 
            
            Consider now $A_2$. %
            Summing up the neighborhood sizes of all agents in~$A_2$ yields at least $2|A_2|$. Since  each red agent in~$P_1$ can contribute to the neighborhood of at most two agents in~$A_1$, but the ones with the smallest and the largest index can contribute to the neighborhood of at most one agent in~$A_2$, we get that there must be at least $|A_2|+1$ red agents in~$P_1$. Hence, labeling $x$ agents red in~$P_1$ yields at most~$x-1$ agents with a majority-red neighborhood in~$P_2$. 
    
            Altogether, it follows that the value of~$\mrn[\opt|x]$ is at most the right-hand side of~\eqref{eq:MI:path:opt-for-odd}. 
            Moreover, a suitable labeling reaching the claimed value exists for each $x \in [\ell]$: as in \Cref{clm:MI:paths:even}, it suffices to label the first~$x$ the agents according to the order $a_2,a_4,\dots,a_{\ell}, a_{\ell-1},a_{\ell-3},\dots,a_1$ as red.   
         \end{claimproof}

     As the majority color must be blue, the maximum number of red agents is $|\agents|/2 -1$ if $|\agents|$ is even, and $\lfloor|\agents|/2| \rfloor$ otherwise. I.e., \Cref{clm:MI:paths:even,clm:MI:paths:odd} imply the statement.
\end{proof}

The observations in the proof of \Cref{thm:MI:path:linear} can be used to solve the case when the social network consists of several disjoint paths, leading to a linear-time algorithm. 

\begin{restatable}%
{theorem}{thmMIpolylinearForest}    
    \label{thm:MI:poly:linearForest}
    If the social network is the union of disjoint paths, \probName{$q$-Majority Illusion} is solvable in linear time.
\end{restatable}
\begin{proof}
    We show that there is a simple, linear-time method to create an optimal labeling~$\rblab$ of the set~$\agents$ of agents in~$G$ optimally, i.e., one that maximizes the number of agents under illusion. Then comparing $\umi[\rblab](\agents)$ with $q \cdot |\agents|$ suffices to decide whether our instance is a yes-instance.

    Our method relies on \Cref{clm:MI:paths:even,clm:MI:paths:odd}. 
    Let $\Peven$ and $\Podd$ denote the set of disjoint paths in~$G$ with an even and odd number of agents, respectively, and let $\calP=\Peven \cup \Podd$. 
    Furthermore, let us adopt the notation $P_1$ and~$P_2$ as defined in the proof of \Cref{thm:MI:path:linear} for each path $P$ in~$G$.

    Let $f$ be a labeling of~$G$ that maximizes $\umi[\rblab](G)$. We can clearly assume that $\RedNo_\rblab(\agents)=\lfloor |(\agents|-1)/2 \rfloor$.
    Let us call a path~$|P|$ in~$G$ \emph{$f$-superb} if it is odd and has $\lfloor |P|/2 \rfloor$ red agents;
    then $\umi(P)=\lceil |P|/2 \rceil$.
    Note that for the set~$\calS_f$ of $f$-superb paths, it must hold that 
    \begin{equation}
    \label{eq:paths:calS-bound}    
    \left\lfloor\frac{|\agents|-1}{2}\right\rfloor 
    =
    \RedNo_\rblab(\agents)
    \leq 
    \sum_{P \in \calS_f}\left\lfloor \frac{|P|}{2} \right\rfloor +
        \sum_{P \in \calP \setminus \calS_f} |P|.
    \end{equation}
    Note also that 
    \begin{align*}
        \umi[f](\agents) &=
        \umi[f]\left(\bigcup_{P \in \calS_f} P\right) + 
        \umi[f]\left(\bigcup_{P \in \calP \setminus \calS_f} P\right)
        \\
        & \leq 
        \sum_{P \in \calS_f} \left\lceil \frac{|P|}{2} \right\rceil
        + 
        \sum_{P \in \calP \setminus \calS_f}  \RedNo_f(P) \\
        & = \sum_{P \in \calS_f} 
        (\RedNo_f(P)+1)
        + 
        \sum_{P \in \calP \setminus \calS_f}  \RedNo_f(P)
        \\
        & = \RedNo_f(\agents)+|\calS_f|
    \end{align*} 
    where the inequality is implies by \Cref{clm:MI:paths:even,clm:MI:paths:odd}.

    We use the following algorithm:
    \begin{enumerate}
        \item We create add to a path family~$\Pi$ paths which will be $f'$-superb for the labeling~$f'$ we construct.
        Initially, $\Pi=\emptyset$, and we add paths from~$\Podd$ to~$\Pi$ one by one, in increasing order of their length, as long as the  condition 
        \begin{equation}
            \label{eq:paths:Pi-bound}
        \left\lfloor\frac{|\agents|-1}{2}\right\rfloor
        \leq
        \sum_{P \in \Pi}\left\lfloor \frac{|P|}{2} \right\rfloor +
        \sum_{P \in \calP \setminus \Pi} |P| 
        \end{equation}
        holds.
        \\
        Note that due to the ordering of the odd paths in which we add them to~$\Pi$, we know that $|\Pi| \geq |\calS_f|$, because $\calS_f$ satisfies~\eqref{eq:paths:calS-bound}.
        \item We label each agent in~$\{P_2:P \in \Pi\}$ as red; let $R$ be number of these agents.
        \item We label $\lfloor (|\agents|-1)/2 \rfloor - R$ agents on paths in~$\calP \setminus \Pi$ as red in a way that yields $\lfloor (|\agents|-1)/2 \rfloor - R$ agents with a majority-red neighborhood in these paths; this can be achieved as shown in \Cref{clm:MI:paths:even,clm:MI:paths:odd}. 
    \end{enumerate} 
    Notice that this way, we can label exactly $\lfloor (|\agents|-1)/2 \rfloor$ agents red, due to~\eqref{eq:paths:Pi-bound}; in particular, the majority color will be blue. Moreover, \Cref{clm:MI:paths:even,clm:MI:paths:odd} guarantee that the number of agents under majority illusion in the obtained labeling~$f'$ is exactly $\RedNo_{f'}(\agents)+|\Pi| \geq \RedNo_{f}(\agents)+|\calS_f|=\umi(G)$. Hence, the labeling~$f'$ is optimal. 
    To decide whether our instance is a yes-instance, it suffices to compare $\umi[f'](\agents)$ with $q \cdot |\agents|$.

    Clearly, the running time of our algorithm is $\Oh{|N|}$.
\end{proof}

Our final result shows that the method we apply to solve \qMI on disjoint paths is not robust with respect to the addition of agents: the presence of only a few additional agents leads to intractability, since the problem is \Wh when parameterized by the distance to disjoint paths of the social network.  

\begin{restatable}%
{theorem}{thmMIWhardddp}
    \label{thm:MI:W1hard:ddp}
    \probName{$q$-Majority Illusion} is \Wh  when parameterized by the distance to disjoint paths $\ddp(G)$ of the social network~$G$.
\end{restatable}
\begin{proof}
    We are going to present a parameterized reduction from \probName{Multidimensional Halfset Sum} when parameterized by the dimension~$d$. Let $(d,\calS)$ be our input instance with $\calS=\{\myvec{s}_1,\dots,\myvec{s}_n\}$. 
    We may assume without that $|\myvec{s}_i|=\ssize$ for some integer $\ssize \in \mathbb{N}$, as otherwise we can create an equivalent instance by appending to each vector~$\myvec{s}_i \in \calS$ a new coordinate with value~$\ssize-|\myvec{s}_i|$ for some large enough integer~$\ssize$, say $\ssize=\max_{i \in [n],j \in [d]} \myvec{s}_i[j]$. 
    Observe that this transformation indeed yields an equivalent instance, since if a set~$\calS'$ of $\frac{n}{2}$ vectors sum up to %
    $\frac{\ssum}{2}$ where $\ssum=\sum_{i \in [n]} \myvec{s}_i$, then the newly appended last coordinates of the vectors in~$\calS'$, as well as in~$\calS \setminus \calS'$, automatically sum up to $\frac{n}{2}\ssize-|\ssum|$.
    We may further assume that $n>d+2$ and that $\ssum[j]>4d$ for each~$j \in [d]$.
    Moreover, by multiplying all vectors in~$\calS$ by~$4$, we ensure that each coordinate of every vector in~$\calS$ (and, hence, $\ssize$) is divisible by~$4$.

    \proofsubparagraph{Construction}
    Let us now define an instance~$I=(N,G,q)$ of \probName{$q$-Majority Illusion} as follows.
    We define the integer $\alpha=n-d-2$. 

    To define the social network~$G$ of our instance~$I$, we create a \emph{short path} for each $i \in [\alpha]$ which is a path $X_i=(a_i^1,b_i^1,a_i^2,b_i^2,a_i^3)$,
    and a \emph{long path}  for each $i \in [n]$ defined as a path~$Y_i=(p_i^1,q_i^1,p_i^2,q_i^2,\dots,p_i^{\ssize},q_i^{\ssize},p_i^{\ssize+1})$.
    We will use the agent sets 
    \begin{equation*}
    \begin{array}{r@{\hspace{2pt}}l@{\hspace{12pt}}r@{\hspace{2pt}}l}
    A_i&=\{a_i^j:j \in [3]\},   & 
    P_i& =\{p_i^j:j \in [\ssize+1]\},\\[3pt]
    B_i&=\{b_i^j:j \in [2]\},  &
    Q_i&=\{q_i^j:j \in [\ssize]\}. \\[3pt]
    \end{array}
    \end{equation*}
    We further let $A=\bigcup_{i \in [\alpha]} A_i$, and we define the sets~$B$, $P$, and~$Q$ analogously.
    
    Additionally, we create a set~$M=\{m_0,m_1,\dots,m_d\}$ of \emph{modulator agents}.
    We connect~$m_0$ with all agents in~$B$ as well as with every other modulator agent. 
    Moreover, we connect~$m_1$ to the first $\myvec{s}_i[1]$ agents of~$Q_i$ on~$Y_i$ (i.e., $q_i^1,\dots,q_i^{\myvec{s}_i[1]}$) for each $i \in [n]$, then connect $m_2$ to the next $\myvec{s}_i[2]$ agents of~$Q_i$ on~$Y_i$, and so on, with each modulator~$m_j$, $j \in [d]$, connected to exactly~$\myvec{s}[j]$ agents from~$Q_i$ (placed ``consecutively'' along~$Y_i$). This way, every agent in~$Q_i$ is connected to exactly one modulator agent, by $|Q_i|=\ssize$. This defines the social network~$G$ over agent set~$N =\Aodd \cup \Aeven \cup M$. Note that the number of agents is $|N|=5\alpha+n(2\ssize+1)+d+1$.
    
    Finally, we set $q$ such that $q|N|=|N|-\frac{n}{2}(\ssize+1)$.

    \proofsubparagraph{Correctness}
    Assume first that there is a red--blue labeling~$\rblab$ of~$G$ with majority color blue that leaves $q|N|$ agents under majority illusion; let us choose~$\rblab$ so as to maximize the number of agents under majority illusion.
    We may assume that 
    \begin{equation}
    \label{eq:ddp:redstotal}
    \RedNo_\rblab(N)=\frac{|N|-1}{2}=3\alpha+n\ssize+d+1    
    \end{equation}
    due to our choice of~$\alpha$.

    Notice that the number of agents under majority illusion in~$A_i$ (or in~$P_i$) solely depends on the labeling of agents in~$B_i$ (or in~$Q_i$, respectively).
    labeling one agent in~$\{b_i^1,b_i^2\}$  red for some $i \in [\alpha]$ yields exactly one agent in~$A_i$ under majority illusion, while labeling both of them achieves this for all three agents in~$A_i$.
    Setting 
    \begin{equation}
    \label{eq:ddp:kBdef}
        \RedNo_\rblab(B)=2\alpha-k_B,
    \end{equation}
    we know that the number of agents under majority illusion in~$A$ is at most \begin{equation}
    \label{eq:ddp:umi-in-A}
        \frac{3\RedNo_\rblab(B)}{2} \leq 3\alpha-\frac{3k_B}{2}.
    \end{equation}
    
    Similarly, if there are $\ell$ red agents in~$Q_i$, then the maximum number of agents under illusion in~$P_i$ is $\ell$ if $\ell< \ssize$ and $\ell+1=\ssize+1$ if $\ell=\ssize$.
    Hence, the number of agents under majority illusion in~$P$ is at most
    \begin{equation}
    \label{eq:ddp:umi-in-P}
        \frac{\ssize+1}{\ssize} \RedNo_\rblab(Q) .
    \end{equation}
        
    Similarly, the number of agents under majority illusion in~$B_i$ (or in~$Q_i$) solely depends on the labeling of agents in~$A_i \cup \{m_0\}$ (or in~$P_i \cup M \setminus \{m_0\}$, respectively).
    In order to ensure that all agents in~$\bigcup_{i \in \alpha} B_i$ are under majority illusion, it is sufficient to label every agent in~$\{a_i^2:i \in [\alpha]\} \cup \{m_0\}$ red.
    Similarly, to ensure that all agents in~$\bigcup_{i \in [n]} Q_i$ are under majority illusion, it is sufficient to label all agents in~$M \setminus \{m_0\} \cup \{p_i^{2j}:j \in [\ssize/2]\}$ as red.
    Thus, we need not more than~$\alpha+\frac{n\ssize}{2}+d+1$ red agents in~$A \cup P \cup M$ to ensure that all agents in~$B \cup Q$ are under majority illusion.  
    In particular, by \eqref{eq:ddp:redstotal}, there are at least $2\alpha+\frac{n\ssize}{2}$ red agents in~$B \cup Q$, implying $\RedNo_\rblab(Q) \geq \frac{n\ssize}{2}$. 

    \begin{claim}
    \label{clm:ddp:modulators-red}
        All modulator agents are red under~$\rblab$.
    \end{claim}
    \begin{claimproof}
        Clearly, we if $m_0$ is not red but $\RedNo_\rblab(A)>0$, then either we can re-label any red agent in~$A$ as blue and set~$m_0$ as red without decreasing the number of agents under majority illusion. If $\RedNo_\rblab(A)=0$, then we can take the first $d+4$ red agents from~$Q$ in the order where $q_i^j$ precedes $q_{i'}^{j'}$ if $i<i'$ or $i=i'$ and $j<j'$ 
        (recall that $\RedNo_\rblab(Q) \geq \frac{n\ssize}{2} > d+4$), 
        re-label them as blue and label~$m_0$ as well as $d+3$ agents of the form~$a_i^2$ as red: this way, at most~$d+4$ agents in~$P$ and possibly $d$ modulators cease to be under illusion, while at least $2(d+3)$ agents in~$B$ will get (newly) under illusion. By our choice of~$\rblab$, this proves that $m_0$ is red.

        To see that each $m_j \in M \setminus \{m_0\}$ is red too, suppose the contrary.
        Let $\hat{P}_j$ agents of the form~$p_i^{2h}$ with two neighbors in~$N_G(m_j)$. If some agent in~$\hat{P}_j$ is red, then we can re-label it as blue and set~$m_j$ as blue instead, without decreasing the number of agents under majority illusion. 
        Otherwise, 
        let $x+1$ be the odd integer in~$\{d+4,d+5\}$. 
        Let us re-label as blue the first $x+1$ red agents from~$Q$ as in the previous paragraph, causing at most~$x+2+d$ agents from~$P \cup M$ to lose their majority illusion, and re-label~$m_j$ as well as $x$ agents of the form~$\hat{P}_j$ as red, increasing the number of agents under majority illusion in~$Q$ by $2x$. 
        Due to $x+2+d < 2x$, this increases the number of agents under majority illusion, contradicting our choice of~$\rblab$ and proving $\rblab(m_j)=\red$.
    \end{claimproof}

    By \Cref{clm:ddp:modulators-red}, we know that all modulators are red. In addition, there are at most $\alpha+\frac{n\ssize}{2}$ red agents in~$A \cup P$ (as that many red agents suffice to make all agents in~$B \cup Q$ to be under majority illusion). 
    Hence, there exists a non-negative integer $k_{A \cup P}$ such that 
    \begin{equation}
    \label{eq:ddp:kAPdef}
        \RedNo_\rblab(A \cup P)= \alpha+\frac{n\ssize}{2}-k_{A \cup P}.
    \end{equation}
    It is clear that every agent in~$B \cup Q$ needs at least one red neighbor in~$A \cup P$, and each agent in~$A \cup P$ is connected to at most two agents in~$B \cup Q$. Thus, the number of agents under majority illusion in~$B \cup Q$ is at most
    \begin{equation}
    \label{eq:ddp:umi-inBQ}
        2\RedNo_\rblab(A \cup P) \leq 2\alpha+n\ssize-2k_{A \cup P}.
    \end{equation}

    Note that by \Cref{clm:ddp:modulators-red} and~\eqref{eq:ddp:redstotal}, we know 
    \begin{equation}
    \label{eq:ddp:reds-per-category}
    \begin{split}
        \RedNo_\rblab(Q) &= \RedNo_\rblab(N)-|M|-
        \RedNo_\rblab(A \cup P) -\RedNo_\rblab(B)\\  
        & =3\alpha+n\ssize
        -
        (\alpha+\frac{n\ssize}{2}-k_{A \cup P})
        -(2\alpha-k_B)        
        \\
        &=\frac{n\ssize}{2}+k_{A \cup P}+k_B
    \end{split}
    \end{equation}
    where we used the definitions of~$k_{A \cup P}$ and of~$k_B$ as given in~\eqref{eq:ddp:kBdef} and~\eqref{eq:ddp:kAPdef}.
    
    Summing up~\eqref{eq:ddp:umi-in-A}, \eqref{eq:ddp:umi-in-P}, and~\eqref{eq:ddp:umi-inBQ} while taking into account the modulator agents and using~\eqref{eq:ddp:reds-per-category}, the total number of agents under majority illusion is at most
    \begin{align}
        \notag
        3\alpha &-\frac{3k_B}{2}+\frac{\ssize+1}{\ssize} \RedNo_\rblab(Q)+         2\alpha+n\ssize-2k_{A \cup P} +|M| \\
        \notag
        &= 5\alpha +n\ssize+\frac{\ssize+1}{\ssize}\left(\frac{n\ssize}{2}+k_{A \cup P}+k_B\right)-\frac{3k_B}{2}\\
        \notag
        & \phantom{= 5\alpha} - 2k_{A \cup P}+d+1 \\
        \notag
        &\leq  5\alpha+\frac{3n\ssize}{2}-k_{A \cup B}-\frac{k_B}{2}+\frac{n}{2}+d+1
        \\
        \label{eq:ddp:no-slack}
        &= q|N|-k_{A \cup B}-\frac{k_B}{2} \leq q|N|.
    \end{align}
    Since there are at least $q|N|$ agents under illusion in~$\rblab$, the following must hold: 
    \begin{enumerate}[label=(\alph*),left=2pt]
        \item \eqref{eq:ddp:no-slack} holds with equality, which  implies~$k_B=k_{A \cup P}=0$ and, hence $\RedNo_\rblab(Q)=\frac{n\ssize}{2}$;
        \item the number of agents under majority illusion in~$P$ must be exactly $\frac{\ssize+1}{\ssize} \RedNo_\rblab(Q)=\frac{n}{2}(\ssize+1)$;
        \item all modulators must be under majority illusion. 
    \end{enumerate}
    Observe that (a) and~(b) together imply that the set of red agents in~$Q$ must be of the form $\bigcup_{i \in S}{Q_i}$ for some ${I^\star \subseteq [n]}$ with $|I^\star|=\frac{n}{2}$. Hence, for each $j \in [d]$, the modulator agent~$m_j$ has exactly $\sum_{i \in I^\star}\myvec{s}_i[j]$ red neighbors. Therefore, (c) requires that 
    \begin{equation}
    \label{eq:ddp:vectorsum}
    \sum_{i \in I^\star}\myvec{s}_i[j] \geq \sum_{i \in [n] \setminus I^\star}\myvec{s}_i[j],    
    \end{equation}
    since $m_j$ has only one red neighbor not in~$Q$ (namely, $m_0$).
    Summing this up for $j \in [d]$, we get
    \[
    \frac{n}{2}\ssize=
    \sum_{j \in [d]} \sum_{i \in I^\star}\myvec{s}_i[j] 
    \geq \sum_{j \in [d]} 
    \sum_{i \in [n] \setminus I^\star}\myvec{s}_i[j]
    = \frac{n}{2}\ssize
    \]
    by $|\myvec{s}|=\ssize$ for each $\myvec{s} \in \calS$. Hence, equality must hold for~\eqref{eq:ddp:vectorsum} for each $j \in [d]$, implying that $\{\myvec{s}_i:i \in I^\star\}$ is a solution for our instance~$I$of \probName{Multidimensional Halfset Sum}.

    \medskip
    For the other direction, suppose $I$ is  yes-instance, i.e., there exists a subset~$\calS'$ of $\calS$ of size~$\frac{n}{2}$ for which $\sum_{\myvec{s} \in \calS'} \myvec{s}=\frac{\ssum}{2}$.
    Let $I^\star=\{i: i \in [n],\myvec{s_i} \in \calS'\}$.
    We define a red--blue labelling~$\rblab$ for~$J$ under which exactly $q|N|$ agents are under majority illusion. 

    First, let us color all modulator agents, all agents in \[M \cup \{a_i^2,b_i^1,b_i^2:i \in [\alpha]\} \cup \{p_i^{2j}:i \in [n],j \in [\ssize/2]\}\] red, as well 
    as all agents in~$\bigcup_{i \in I^\star} Q_i$. By our previous arguments, this means that all agents in~$A \cup B \cup Q$, as well as all agents in~$\bigcup_{i \in I^\star} P_i$ are under majority illusion.
    Moreover, since~$\calS'$ is a solution, we know that \eqref{eq:ddp:vectorsum} holds for each $j \in [d]$ with equality; since $m_0$ is red and is the only neighbor of~$m_j$ not in~$Q$, we get that $m_j$ is also under majority illusion.
    This clearly holds for $m_0$ as well, who has only red neighbors. Hence, the set of agents under majority illusion is exactly $N \setminus \bigcup_{i \in [n] \setminus I^\star}P_i$, a set of size exactly~$q|N|$. 
    This shows the correctness of our reduction.
\end{proof}

\section{Bounded Distance to Cluster Graph}

Finally, let us turn to the parameterization by the distance of the social network to cluster graphs; that is, we assume that removing a few agents results in a collection of cliques; however, these cliques can have unbounded size. Again, we are not the first considering this parameter in the context of social networks; see, e.g., \cite{EibenGO2018,FioravantesKKMO2025a,FioravantesGMS2026}.

\begin{theorem}
\label{thm:MI:XP:dcg}
    \probName{$q$-Majority Illusion} is in \XP when parameterized by~$\operatorname{dcg}(G)$, the distance to cluster graphs of the social network~$G$.
\end{theorem}
\begin{proof}
	Let $M$ be a set of agents of size $k=\operatorname{dcg}(G)$ such that $G - M$ is a disjoint union of cliques $K_1,\ldots,K_\ell$. This time, our algorithm first guesses the labeling~$f_M$ of~$M$ and then extends~$f_M$ using a dynamic-programming approach over the cliques $K_1,\ldots,K_\ell$. 
    Let $K_{j \to \ell}={\vset{K_j}\cup \cdots \cup \vset{K_\ell}}$.
	
	For every clique $j\in[\ell]$, we have a dynamic programming table~$\DP_j$ containing values $\DP_j[R,\myvec{r}]$, where
    \begin{itemize}
        \item $R\in[|\agents|/2]$ is the number of red agents in~$K_{j \to \ell}$, and
        \item $\myvec{r} = (r_i)_{i\in M}$, where $r_i \in [|\agents|]_0$ is the number of red neighbors of agent $i\in M$ in~$K_{j \to \ell}$.
    \end{itemize}
    The value $\DP_j[R,\myvec{r}]$ stores the maximum number of agents in $K_{j \to \ell}$ 
    with a majority-red neighborhood
    over all possible extensions~$f$ of~$f_M$ such that $\RedNo_f(K_{j \to \ell}) = R$ and for every $i\in M$ it holds that $\RedNo_f(\neighb_{K_{j \to \ell}}(i) ) = r_i$.

    Before we formally define the computation, we show an auxiliary lemma proving that, for each clique $K_j$, the number of relevant labelings can be bounded in terms of our parameter. Specifically, we show that if we swap the labels of two twins, the number of agents with majority-red neighborhood remains the same.

    \begin{claim}%
        \label{clm:MI:XP:dcg:twinsExchange}
        Let $f$ be an extension of $f_M$ and $u,v$ be a pair of agents in~$K_j$ for some $j \in [\ell]$ such that $f(u) = \red$, $f(v) = \blue$, and $u$ and~$v$ are twins, i.e., $\neighb(u)\setminus\{v\} = \neighb(v)\setminus\{u\}$. Then, 
        $\mrn[f](N)=\mrn[f'](N)$
        where
        \[
            f'(w) = \begin{cases}
                \red  & \text{if } w = v,\\
                \blue & \text{if } w = u\text{, and}\\
                f(w)  & \text{otherwise.}
            \end{cases}
        \]
    \end{claim}
    \begin{claimproof}
        Let $N_{uv}=(\neighb(u)\cup \neighb(v)) \setminus \{u,v\}$. Each agent that is not in~$N_{uv}$ is not affected by the swap of labels. For each agent in $N_{uv}$, the number of red neighbors also remains the same, since $\neighb(u)\setminus\{v\} = \neighb(v)\setminus\{u\} = N_{uv}$. Finally, it may happen that $v$ has a majority-red neighborhood in $f$, but not in $f'$. 
        In this case, $\RedNo_f(N_{uv})=\BlueNo_f(N_{uv})$, and thus
        $u$ does not have majority-red neighborhood in~$f$, but it does so in~$f'$. This proves the claim.
    \end{claimproof}

    Let $K_j$, $j\in[\ell]$, be a clique in~$G$. Observe that the agents in $K_j$ can be partitioned into $\Oh{2^k}$ classes of twins called \emph{types}. Moreover, by \Cref{clm:MI:XP:dcg:twinsExchange}, swapping the labels of two twin agents does not change the number of agents with a majority-red neighborhood. Hence, we say that two labelings~$f$ and~$f'$ are \emph{equivalent} if they have the same number of red agents of each type. Consider the resulting equivalence classes of all labelings of~$K_j$, and let $\mathcal{F}_j$ contain exactly one labeling from each such class. By the previous arguments, it holds that $|\mathcal{F}_j| \in |\agents|^\Oh{2^k}$. Moreover, we use $\mathcal{F}_j^R$ to denote the set of all  labelings in $\mathcal{F}_j$ with exactly $R$ red agents.

    The base case of the computation is when $j = \ell$. To compute $\DP_\ell[R,\myvec{r}]$, we check if we can satisfy the requirement prescribed by the pair $(R,\myvec{r})$. Specifically, we iterate over all labelings of~$K_j$ in~$\mathcal{F}_j$ that label exactly~$R$ agents red, and for each of them we check  
    that they satisfy that each modulator agent $i\in M$ has exactly $r_i$ red neighbors. For all such labelings, we compute the number of agents with majority-red neighborhoods, and store the maximum value in $\DP_j[R,\myvec{r}]$.  Formally, the computation is as follows:\footnote{Here and elsewhere, we let the maximum of some value taken over an empty set be defined as $-\infty$.}
    \[
        \DP_\ell[R,\myvec{r}] = 
        \max\limits_{\substack{
                f\colon \mathcal{F}_\ell^{R}
                \text{ s.t. }
                \forall i\in M\colon \\\RedNo_f(
                 \neighb_{\vset{K_\ell}}(i) )
                =  r_i
            }} \!\!\!\!\!\! \mrn[f \cup f_M](\vset{K_\ell})
    \]

    \def\fopt{f^{\textup{opt}}}
    For $j \in [1,\ell-1]$, we compute $\DP_j[R,\myvec{r}]$ as follows. %
    We iterate over all possible values for the number~$R'$ of red agents in~$K_j$ (from among all $R$ red agents in~$K_{j \to \ell}$),
    and over all labelings of~$K_j$ in~$\mathcal{F}_j$ that label exactly~$R'$ agents in~$K_j$ red. 
    For all values of~$R'$ and for all labelings in~$f \in  \mathcal{F}^{R'}_j$, we compute the maximum number of agents with a majority-red neighborhood in~$K_{j \to \ell}$ under some suitable labeling respecting our choices for~$R'$, $f$, and~$(R,\myvec{r})$. 
    This can be done by using previously computed values from  table~$\DP_{j+1}$. More precisely, we need to ensure that 
    (i) there are $R-R'$ red agents in~$K_{j+1 \to \ell}$, and that 
    (ii) each modulator agent~$i \in M$ has exactly $r_i$ red neighbors in~$K_{j\to \ell}$, that is, 
    $i$ has $r_i-\RedNo_f(\neighb_{\vset{K_j}}(i))$ red neighbors in~$K_{j+1 \to \ell}$.
    The maximum number of agents in~$K_{j+1 \to \ell}$ with a majority-red neighborhood under such a labeling is stored in~$\DP_j[R-R',\myvec{r'}]$ where $\myveccomp{r'}_i=r_i-\RedNo_{f}(\neighb_{\vset{K_j}}(i))$ for $i \in M$.
    Formally, the computation is defined as follows.
    \begin{align*}
        & \DP_j[R,\myvec{r}] = \max_{R'\in[\min\{|\vset{K_j}|,R\}]_0} \max_{f\in\mathcal{F}_j^{R'}} \\ 
        & \phantom{a} \mrn[f \cup f_M](\vset{K_j}) + \DP_{j+1}[R-R',(r_i - \RedNo_f(\neighb_{\vset{K_j}}(i)))_{i\in M}]\,.
    \end{align*}

    Once the dynamic programming table for clique $K_1$ is correctly computed, we check whether there exists a pair $(R,\myvec{r})$ for some $R\in[\lfloor\frac{|\agents|-1}{2}\rfloor - \RedNo_{f_M}(M)]_0$ and $\myvec{r} \in [|\agents|]_0^k$ so that 
    \begin{multline*}
    \label{eq:DP:final-umi}
        \DP_1[R,\myvec{r}]  + |\{ i \in M \colon r_i + \RedNo_{f_M}(\neighb_M(i)) > \frac{\deg(i)}{2}\}|\\
        > q \cdot|\agents|.
    \end{multline*}
    If so, we return \Yes; otherwise, we return \No. 
    Note that our choice for~$R$ ensures that blue is the majority color, while our second condition guarantees that the number of agents under majority illusion is at least~$q|\agents|$. Thus, the correctness of our algorithm is implied by the following claim for $j=1$.

    \begin{restatable}%
    {claim}{clmMIdcgDPcorrect}
    \label{clm:MI:dcg:DP-correct}
        For all values of $j \in [\ell]$, $R \in [\lfloor|\agents|/2\rfloor]_0$, and $\myvec{r}=(r_i)_{i \in M}$ with $r_i \in [|\agents|]_0$, we have
        $\DP_j[R,\myvec{r}]=w$ (where $w \in \mathbb{N}$)  
        if and only if
        there exists a labeling~$f$ of~$G$ that extends~$f_M$ and satisfies the following conditions:
        \begin{enumerate}[label=(\alph*),left=2pt]
            \item $\RedNo_f(K_{j \to \ell})=R$,
            \item $\RedNo_f(\neighb_{K_{j \to \ell}}(i))=r_i\,$ for each $i \in M$, and
            \item $\mrn[f](K_{j \to \ell})=w$.
        \end{enumerate}
    \end{restatable}
    \begin{claimproof}
        We prove the statement by induction on~$j$, for decreasing values of~$j$ from~$\ell$ to~$1$.
    
        Consider the case for $j=\ell$, and recall that 
        \begin{equation}
            \label{eq:MI:dcg:DP:correctness:last}
            \DP_\ell[R,\myvec{r}] = 
                \max\limits_{\substack{
                    f\colon \mathcal{F}_\ell^{R}
                    \text{ s.t. }
                    \forall i\in M\colon \\\RedNo_f(
                     \neighb_{\vset{K_\ell}}(i) )
                    =  r_i
                }} \!\!\!\!\!\! \mrn[f \cup f_M](\vset{K_\ell}), 
        \end{equation}
    
        First, assume that $f$ is a labeling of~$G$ satisfying conditions (a)--(c) of the claim for some $R$ and $\myvec{r}$. 
        By the definition of equivalent labelings and of~$\mathcal{F}_{\ell}^R$, it is clear that there must be a labeling~$f'$ in~$\mathcal{F}^R_{\ell}$ that is equivalent to the restriction of~$f$ to~$K_\ell$ and satisfies $\RedNo_{f'}(\neighb_{K_\ell}(i))=\myveccomp{r}_i$ for each modulator~$i \in M$. Then in the expression~\eqref{eq:MI:dcg:DP:correctness:last} for computing~$\DP_\ell[R,\myvec{r}]$, the maximum is taken over a non-empty subset, and in particular, the value of~$\DP_\ell[R,\myvec{r}]$ will be set to at least the number of agents with a majority-red neighborhood in~$K_\ell$ under~$f'$, that is, to $\mrn[f'](K_\ell)$ which, in turn, equals $\mrn[f](K_\ell)$ by the definition of equivalence of labelings. 
        
        Suppose now that $\DP_\ell[R,\myvec{r}]=w$; 
        we show that there exist a labeling~$f$ of~$G$ satisfying  conditions (a)--(c) of the claim. Let $f \in \mathcal{F}^R_{\ell}$ for which $\mrn[f \cup f_M](K_\ell)=w$ and also 
        $\RedNo_f(\neighb_{K_\ell} (i))=\myveccomp{r}_i$ for each $i \in M$ holds. Let $f'$ be an arbitrary extension of~$f \cup f_M$ over~$\agents$. 
        Then conditions~(a)--(c) clearly hold for~$f'$ by our choice of~$f'$.
    
        \smallskip
        Consider now the case $j \in [\ell-1]$, 
        and recall that 
        \begin{equation}
            \label{eq:MI:dcg:DP:correctness:gen}
            \begin{split}
                & \DP_j[R,\myvec{r}] = \max_{R'\in[\min\{|\vset{K_j}|,R\}]_0} \max_{f\in\mathcal{F}_j^{R'}} \\ 
                & \phantom{a} \mrn[f \cup f_M](\vset{K_j}) + \DP_{j+1}[R-R',(r_i - \RedNo_f(\neighb_{\vset{K_j}}(i)))_{i\in M}]\,.
            \end{split}
        \end{equation}
    
        Assume first that $f$ is a labeling of~$G$ satisfying conditions (a)--(c) for some $R$ and~$\myvec{r}$. We aim to show that $\DP_j[R,\myvec{r}] \geq \mrn[f](K_{j \to \ell})$. 
        Let $R'=\RedNo_f(K_j)$; then we have $R' \leq \min\{|K_j|,R\}$. 
        Moreover, let $f'$ %
        be the restriction of~$f$ to~$K_j$.
        Observe that by its definition, $\mathcal{F}_j^{R'}$ must contain some labeling $\hat{f}$ of~$K_j$ that is equivalent to~$f'$.
        Therefore, the maximum in the expression~\eqref{eq:MI:dcg:DP:correctness:gen} for~$\DP_j[R,\myvec{r}]$ is taken over a non-empty set. To lower-bound the maximum value obtained in~\eqref{eq:MI:dcg:DP:correctness:gen}, 
        first note that  
        $\mrn[\hat{f} \cup f_M](K_j)=\mrn[f](K_j)$. Second, note that
        \begin{itemize}
            \item $\RedNo_f(K_{j+1 \to \ell})=R-R'$ and 
            \item $\RedNo_f(\neighb_{K_{j+1 \to \ell}}(i))=r_i - \RedNo_f(\neighb_{K_j}(i))$ for each $i \in M$. 
        \end{itemize}
        Hence, if $w'$ denotes the number of agents with majority-red neighborhood in~$K_{j+1 \to \ell}$ under~$f$, then the induction hypothesis yields that $\DP_{j+1}[R-R',(r_i-\RedNo_f(\neighb_{K_j}(i)))_{i \in M}]$ has value at least~$w'$.
        Thus, the right-hand side of~\eqref{eq:MI:dcg:DP:correctness:gen} is at least $\mrn[f](K_j)+w'=\mrn[f](K_{j \to \ell})=w$, as desired.
    
        Suppose now that $\DP_j[R,\myvec{r}]=w$ for some integer~$w$. We are going to construct a labeling of~$G$ satisfying conditions (a)--(c). Let $R' \in [\min\{|K_j|,R\}]$ and $f \in \mathcal{F}_j^{R'}$ be the values under which the right-hand side of~\eqref{eq:MI:dcg:DP:correctness:gen} takes its maximum (namely, $w$) in the computation for~$\DP_j[R,\myvec{r}]$.
        Let $w'$ denote the value of~$\DP_{j+1}[R-R',(r_i-\RedNo_f(\neighb_{K_j}(i)))_{i \in M}]$;
        then we must have 
        $w=\mrn[f \cup f_M](K_j)+w'$, due to the computation in~\eqref{eq:MI:dcg:DP:correctness:gen}.
        By induction, we know that there exists a labeling~$f'$ of~$G$ extending~$f_M$ for which 
        \begin{enumerate}[label=(\alph*'),left=2pt]
            \item $\RedNo_{f'}(K_{j+1 \to \ell}) = R-R'$,
            \item $\RedNo_{f'}(\neighb_{K_{j+1 \to \ell}}(i))=r_i-\RedNo_f(\neighb_{K_j}(i))$ for all $i \in M$, and
            \item $\mrn[f'](K_{j+1 \to \ell})= w'$.
        \end{enumerate}
        Define $\hat{f}$ as an arbitrary labeling of~$G$ that coincides with~$f$ over~$K_j$, with $f'$ over~$K_{j+1 \to \ell}$, and with~$f_M$ over~$M$.
        First, there are $(R-R')$ red agents in~$K_{j+1 \to \ell}$ by~(a') and $R'$ red agents in~$K_j$ under~$\hat{f}$ (because $f \in \mathcal{F}_j^{R'}$), so (a) holds for~$\hat{f}$.
        Second, each modulator agent~$i \in M$ has exactly $r_i$ red neighbors in~$K_{j \to \ell}$ due to~(b'). 
        Third, the number of agents with majority-red neighborhood under~$\hat{f}$ is $w'$ in~$K_{j+1 \to \ell}$ due to~(c'),
        and thus $\mrn[\hat{f}](K_{j \to \ell})=w'+\mrn[f \cup f_M](K_j)=w$. Hence, $\hat{f}$ fulfills conditions (a)--(c), as required. This shows the claim for $j \in [\ell-1]$.
    \end{claimproof}
    
    \proofsubparagraph{Running Time} Each cell in each table $\DP_j$ can be computed in $n^\Oh{2^k}$ time where $n=|\agents|$, as we exhaustively try all labelings in $\mathcal{F}_j^R$. %
    There are $n\cdot n^k \in n^\Oh{k}$ different cells for every table, and overall $\Oh{n}$ different tables. That is, the dynamic programming procedure runs in $n^\Oh{2^k}$ time, and requires $n^\Oh{k}$ space, as it is enough to store tables only for two consecutive cliques~$K_j$ and $K_{j+1}$. As we compute the table for each of the $2^\Oh{k}$ possible labelings~$f_M$ for~$M$, the running time of the algorithm is $2^\Oh{k}\cdot n^\Oh{2^k}$.
\end{proof}

The next result shows that the running time of the algorithm in \Cref{thm:MI:XP:dcg} cannot be improved to obtain an FPT algorithm with parameter~$\dcg(G)$, since the problem is \Wh. 
Its proof relies on a reduction from the \probName{Multidimensional Subset Sum} problem whose input consists of an integer~$d$, a set $\calS=\{\myvec{s}_1,\dots,\myvec{s}_n\}$ of $d$-dimensional vectors from $\mathbb{N}^d$, and a target vector $\myvec{t} \in \mathbb{N}^d$; the task is to decide whether there exists a subset $\calS' \subseteq \calS$ for which $\sum_{\myvec{s} \in \calS'} \myvec{s} = \myvec{t}$. This problem is known to be %
\Wh with parameter~$d$ even if its input is encoded in unary~\cite{GanianKO2021}.

\begin{restatable}%
{theorem}{thmMIWharddcg}
\label{thm:MI:W1hard:dcg}
    \probName{$q$-Majority Illusion} is \Wh when parameterized by the distance to cluster graph~$\operatorname{dcg}(G)$ of the social network~$G$.
\end{restatable}
\begin{proof}
    We present a parameterized reduction from the unary-encoded version of \probName{Multidimensional Subset Sum}; let $I=(d,\calS,\myvec{t})$ be our input instance with $\calS=\{\myvec{s}_1,\dots,\myvec{s}_n\}$. 
    Observe first that we may assume w.l.o.g.\ that $d$ is odd, as otherwise we can increase~$d$ by one, appending a coordinate with value~$0$ to each vector in $\calS \cup \{\myvec{t}\}$. 
    Second, we may also assume that each coordinate of $\myvec{t}$ is positive:
    this can be achieved by adding a sufficiently large vector $\myvec{v} \in \mathbb{N}^d$ satisfying $\myvec{v} > \sum_{\myvec{s} \in \calS} \myvec{s}$) to the set~$\calS$ and also to the target vector~$\myvec{t}$, this way ensuring that any solution will contain $\myvec{v}$.
    Third, we can also assume that each coordinate of~$\myvec{v}$ is divisible by~$2d+6$ (in particular,
    $|\myvec{v}|\geq 2d+6$) for each $\myvec{v} \in \calS \cup \{t\}$, as otherwise we can simply multiply each vector in $\calS \cup \{\myvec{t}\}$ by~$2d+6$.
    Fourth, we may additionally assume that not only $\myvec{t} \leq  \sum_{\myvec{s} \in \calS} \myvec{s}$ but also $2\myvec{t} < \sum_{\myvec{s} \in \calS} \myvec{s}$ holds, as otherwise we can add a vector $\myvec{t}'>\myvec{t}$ to~$\calS$ which can never be contained in a solution.
    We create an instance $J=(N,G,q)$ of \probName{$q$-Majority Illusion} as follows.

    \proofsubparagraph{Construction}
    Let us introduce some notation. We will denote by $\ssum =\sum_{i=1}^n \myvec{s}_i$ the sum of all input vectors, and we define integers $\alpha=2|\ssum|+d+1$ and $n_\alpha=|\ssum|-|\myvec{t}|-\frac{d+2}{2}$.
    
    For each $i \in [n_\alpha]$, we define a clique $A_i$ over exactly $2\alpha$ agents. The agents in~$A_i$ will not be connected to any agent outside~$A_i$ except for $i=1$.
    Next, for each $i \in [n]$, we define a clique $K_i$ of size~$2|\myvec{s}_i|$ over agent set $\bigcupdot_{j=0}^d K_i^j$, where $|K_i^0|=|\myvec{s}_i|$ and $|K_i^j|=\myvec{s}_i[j]$ for each $j \in [d]$. Finally, we add a set $M=\{m_0,m_1,\dots,m_d\}$ of \emph{modulator agents}. We connect $m_0$ to each agent in~$K:=\bigcup_{i \in [n],j \in [d]} K_i^j$. 
    Next, for each $j \in [d]$ we connect $m_j$ to each agent in~$\bigcup_{i \in [n]} K_i^j$, as well as to exactly $\ssum[j]-2\myvec{t}[j]+1$ agents from~$A_1$ in a way that no two agents in~$A_1$ are connected to more than one modulator agent. 
    Note that this is possible, because 
    $\sum_{j \in [d]}(\ssum[j]-2\myvec{t}[d]+1)=|\ssum|-2|\myvec{t}|+d>0$.
    This finishes the definition of the social network~$G$ over the set~$N$ of agents;
    note that $|N|=2\alpha \cdot n_{\alpha} + 2|\ssum|+d$.

    We set $q$ such that $q|N|=2\alpha \cdot  n_\alpha + 2|\myvec{t}|+d$.

    Observe that the constructed social network has cluster deletion number at most~$d+1$, because deleting all modulator agents from~$G$ leaves a collection of disjoint cliques. Thus, the presented reduction is a parameterized one.

    \proofsubparagraph{Correctness}
    We are going to show that $I$ is a yes-instance of \probName{Multidimensional Subset Sum} if and only if $J$ is a yes-instance of \probName{$q$-Majority Illusion}.

    Assume first that there are at least $q|N|$ agents under majority illusion for some red--blue labeling~$\rblab$ (with blue being the majority color). We need the following observation, obtained by simple counting:
    \begin{claim}\label{clm:even-clique}
        Let $C$ be a clique of even size whose each agent has at most one red neighbor outside~$C$. If it contains fewer than~$|C|/2$ red agents, then no agent in~$C$ is under majority illusion; if $C$ contains exactly $|C|/2$ red agents, then no red agent is under majority illusion.    
    \end{claim}
    By \Cref{clm:even-clique} applied to the cliques $A_i$, $i \in [n_\alpha]$,  we know that either each of~$A_i$ contains at least $\alpha+1$ red agents, or the number of agents under majority illusion is at most 
    \begin{align*}
        2\alpha \cdot(n_\alpha-1)+\alpha+2|\ssum|+d=2\alpha \cdot n_\alpha -1 < q|N|.
    \end{align*}
    Thus, there must be at least $\alpha+1$ red agents in each clique~$A_i$, ensuring that all agents in these cliques are under majority illusion.

    Next, we show that 
    \begin{equation}
        \label{eq:reds-in-cliques}
        \RedNo_\rblab(K \cup M)\leq |\myvec{t}|+d+1
    \end{equation}
    as otherwise 
    \begin{align*}
        \RedNo_\rblab(N) &\geq (\alpha+1) n_{\alpha}+|\myvec{t}|+d+2 \\
        &= 
    (\alpha-1)n_\alpha + 2|\ssum|-|\myvec{t}| \geq \BlueNo_\rblab(N)
    \end{align*}
    where the equality follows from our choice of~$n_\alpha$, and the last inequality from $|K \cup M|=2|\ssum|+d+1$.

    Let us say that some clique $K_i$, $i \in [n]$, is \emph{selected}, \emph{half-selected}, or \emph{rejected} by~$\rblab$ if $\RedNo_\rblab(K_i)$ is at least~$|K_i|/2=|\myvec{s}_i|$, exactly $|\myvec{s}_i|-1$, or less, respectively. By an analog of \Cref{clm:even-clique} for cliques where some agents can have two neighbors outside the clique, we get that if $K_i$ is rejected, then no agent in it is under majority illusion; if $K_i$ is half-selected, then at most~$|\myvec{s}_i|$ agents in~$K_i$ can be under majority illusion, whereas this is possible for all agents in a selected clique. 
    
    \begin{claim}
    \label{clm:clique-selection}
        There are exactly~$2|\myvec{t}|$ agents under majority illusion in~$K$, all in selected cliques of total size~$2|\myvec{t}|$.
    \end{claim}
    \begin{claimproof}
    We first show that no clique in~$K$ is half-selected by~$\rblab$. 
    Suppose for the sake of contradiction that $K_i$ is such a clique. 
    Then $\RedNo_\rblab(K_i)=|\myvec{s}_i|-1$ and the number of agents in~$K_i$ under majority illusion is $|\myvec{s}_i| \leq 2|\myvec{s}_i| - (2d-6)$ by our assumption on~$|\myvec{s}_i|$. 
    Let $\beta$ denote the number of red vertices in~$K \setminus K_i$; recall that by \eqref{eq:reds-in-cliques} we have $\beta+|\myvec{s}_i|-1 \leq |\myvec{t}|+d+1$. Then there are at most $2\beta$ agents under majority illusion in~$K \setminus K_i$.
    Thus,
    the number of agents in $K$ under majority illusion is at most
    \begin{align*}
        2\beta+|\myvec{s}_i| &\leq
        2\beta+2|\myvec{s}_i|-(2d+6)  \\
        &\leq 
        2(|\myvec{t}|+d+2)-2d-6 
        =2|\myvec{t}|-2
    \end{align*}    
    which yields that the total number of agents under majority illusion is $2\alpha \cdot n_\alpha+2|\myvec{t}|-2+(d+1)<q|N|$, a contradiction. 
    
    Thus, all cliques in~$K$ are either selected or rejected. Hence, the set of agents under majority illusion in~$K$ is exactly the agent set of all selected cliques. Since there are at most~$|\myvec{t}|+d+1$ agents in~$K$ and each $|\myvec{s}_i|$, $i \in [n]$, as well as $|\myvec{t}|$ are divisible by $2d+6$, we know that the maximum total size of selected cliques is $2|\myvec{t}|$.   
    Note that if the the total size these cliques is less than~$2|\myvec{t}|$ or not all agents in them are under majority illusion, then the number of such agents in~$K$ is at most $2\myvec{t}-(2d+6)$, and we again reach a contradiction because this yields less than~$q|N|$ such agents in total.
    This proves %
    the claim.
    \end{claimproof}

    Let $\calS'$ denote the set of vectors~$\myvec{s}_i \in \calS$ for which $K_i$ is selected by~$\rblab$. By \Cref{clm:clique-selection}, we know that 
    \begin{equation}
    \label{eq:tsum-reached}
        \sum_{\myvec{s}_i \in \calS} |\myvec{s}_i|=|\myvec{t}|.
    \end{equation}

    Observe that agent~$m_0$ is not under majority illusion, as it has $2|\ssum|$ neighbors but at most $|\myvec{t}|+d+1<2|\myvec{t}|<|\ssum|$ holds by our assumptions.
    Hence, all other modulator vertices must be under majority illusion, resulting in exactly $2\alpha \cdot n_\alpha+2|\myvec{t}|+d=q|N|$ such agents in total. This means that each modulator agent~$m_j$, $j \in [d]$ needs at least~$\myvec{t}[j]$ red neighbors as otherwise $m_i$ has at most 
    \[
    \ssum[j]-2\myvec{t}[j]+1+(\myvec{t}[j]-1)=\ssum[j]-\myvec{t}[j]
    \]
    red neighbors but at least $\ssum-\myvec{t}[j]+1$ blue neighbors. 
    Using that $\myvec{t}[j]$ as well as $|K_i^j|$ for each $i \in [n]$  is divisible by $2d+6$, we get that $m_j$ must have at least $\myvec{t}[j]$ neighbors in selected cliques. 
    This means that 
    \begin{equation}
    \label{eq:solution-bound}
        \sum_{\myvec{s}_i \in \calS'} \myvec{s}_i[j] \geq \myvec{t}[j]
    \end{equation}
    holds for each $j \in [d]$. Taking into account \eqref{eq:tsum-reached} we get that 
    equality must hold for each $j \in [d]$ in~\eqref{eq:solution-bound}. Hence, we obtain $\sum_{\myvec{s} \in \calS'} \myvec{s}=\myvec{t}$, proving that $\calS'$ is a solution to our instance~$I$
    of \probName{Multidimensional Subset Sum}. 

    \medskip
    For the other direction, let us now assume that $\calS' \subseteq \calS$ is a solution for~$I$. We create a red--blue labeling~$\rblab$ with majority color blue as follows. 
    We color exactly $\alpha+1$ agents red in each~$A_i$, $i \in [n_\alpha]$; we take care to color red all agents in~$A_i$ connected to any modulator agent (it is easy to check that there are not more than $|A_i|/2$ such agents).
    This yields ${2\alpha \cdot n_\alpha}$ agents under majority illusion in these cliques. 
    Next, we color all modulator vertices as well as all agents in~$\bigcup\{K_i^j:\myvec{s}_i \in \calS',j \in [d]\}$ red. This way, each modulator agent~$m_j$, $j \in [d]$, has exactly $\myvec{t}[j]$ red neighbors in~$K$ and is therefore under majority illusion. 
    Finally, all agents in~$K_i$ with $\myvec{s}_i \in \calS'$ are under majority illusion, because in each such clique exactly half of the agents are red, and moreover, each red agent has two red neighbors (namely, two modulator agents)    
    outside the clique. This yields in total  exactly $2\alpha \cdot n_\alpha + 2|\myvec{t}|+d= q|N|$ agents under majority illusion. 

    It remains to check that blue is indeed the majority color under~$\rblab$. Note that 
    \begin{align*}
        \RedNo_\rblab(N) &= (\alpha+1) n_{\alpha}+|\myvec{t}|+d+1 \\
        &= 
    (\alpha-1)n_\alpha + 2|\ssum|-|\myvec{t}|-1 = \BlueNo_\rblab(N)-1
    \end{align*}
    which proves the correctness of the reduction.
\end{proof}

We conclude with an \FPT algorithm for the parameterization by the cluster edge deletion number, that is, the minimum number of edges we need to remove to obtain a disjoint union of cliques. %

\begin{theorem}\label{thm:MI:FPT:cdn}
    \probName{$q$-Majority Illusion} is in \FPT when parameterized by the cluster edge deletion number $\cdn(G)$ of the social network $G$.
\end{theorem}
\begin{proof}
    Let $F$ be a set of $k=\cdn(G)$ edges such that $G-F$ is a disjoint union of cliques $K_1,\ldots,K_\ell$, and let $W$ be the set of all endpoints of the edges in~$F$; note that $|W|\le 2k$. %
    Since each $K_j$ is a clique that contains no edge of~$F$, every edge of~$F$ joins two \emph{distinct} cliques. We record two immediate consequences:
    \begin{enumerate}[label=(\roman*),left=2pt]
        \item every agent in $\agents\setminus W$ has all of its neighbors inside its own clique; and
        \item for every agent $w\in W$, each neighbor of~$w$ lying outside its own clique is an endpoint of an edge of~$F$, and hence belongs to~$W$.
    \end{enumerate}

    As in the proof of \Cref{thm:MI:XP:dcg}, our algorithm first guesses the labeling~$f_W$ of~$W$ (there are at most~$2^{2k}$ of them) and then extends~$f_W$ by dynamic programming over the cliques $K_1,\ldots,K_\ell$. The decisive difference lies in property~(ii): once~$f_W$ is fixed, the number of red and blue neighbors that any agent $w\in W$ has \emph{outside} its own clique is a constant determined by~$f_W$. Thus, the status of every agent---including those in~$W$---is determined locally, within a single clique, and we no longer need the vector~$\myvec r$ recording the red neighbors of the modulator agents across cliques. A single counter for the total number of red agents suffices; this is exactly what turns the \XP{} algorithm into an \FPT{} one.

    Let $K_{j \to \ell}=\vset{K_j}\cup \cdots \cup \vset{K_\ell}$. For every clique $j\in[\ell]$ we maintain a table~$\DP_j$  which stores a value  $\DP_j[R]$ for each   $R\in[\lfloor|\agents|/2\rfloor]_0$, interpreted as the number of red agents in~$K_{j \to \ell}$. The value $\DP_j[R]$ stores the maximum number of agents in $K_{j \to \ell}$ with a majority-red neighborhood over all extensions~$f$ of~$f_W$ with $\RedNo_f(K_{j \to \ell})=R$.

    \proofsubparagraph{Relevant Labelings of a Clique} Fix a clique $K_j$. By property~(i), all agents in $\vset{K_j}\setminus W$ %
    are pairwise twins. %
    By \Cref{clm:MI:XP:dcg:twinsExchange}, swapping the labels of two such %
    twins preserves the number of agents with a majority-red neighborhood; moreover, since the swap affects only agents in $\vset{K_j}$, it preserves this number \emph{within} $K_j$ as well. Hence, the number of agents in~$K_j$ with a majority-red neighborhood depends only on~$f_W$ and on the number of red agents in~$K_j$. Accordingly, for each feasible red count~$R'$ we let $\mathcal{F}_j^{R'}$ contain a single representative labeling of the %
    agents $\vset{K_j}\setminus W$ obtained by coloring any $R'-\RedNo_{f_W}(W\cap\vset{K_j})$ of them red; this set is empty (and the corresponding maximum is $-\infty$) unless
    \[
        \RedNo_{f_W}(W\cap\vset{K_j}) \;\le\; R' \;\le\; \RedNo_{f_W}(W\cap\vset{K_j}) + |\vset{K_j}\setminus W|.
    \]
    In particular $|\mathcal{F}_j^{R'}|\le 1$, so $K_j$ has only $\Oh{|\vset{K_j}|}$ relevant labelings in total---in contrast to the bound $|\agents|^{\Oh{2^k}}$ of \Cref{thm:MI:XP:dcg}.

    \proofsubparagraph{Computation} The base case $j=\ell$ is computed as
    \[
        \DP_\ell[R] = \max_{f\in\mathcal{F}_\ell^{R}} \mrn[f \cup f_W](\vset{K_\ell}),
    \]
    where each $\mrn[f\cup f_W](\vset{K_\ell})$ is evaluated directly: an %
    agent~$a$ in~$K_{\ell} \setminus W$ sees only $\vset{K_\ell}\setminus\{a\}$, and an agent $w\in W\cap\vset{K_\ell}$ sees, in addition, its red and blue neighbors outside~$K_\ell$, whose number is fixed by~$f_W$ thanks to property~(ii). 
    
    For $j\in[\ell-1]$, we have
    \begin{multline}
        \label{eq:cdn-DP}
        \DP_j[R]  = %
        \\
        \max_{R'\in[\min\{|\vset{K_j}|,R\}]_0} \ 
        \max_{f\in\mathcal{F}_j^{R'}} \Big( \mrn[f\cup f_W](\vset{K_j}) + \DP_{j+1}[R-R'] \Big).
    \end{multline}

    Once $\DP_1$ is computed, we return \Yes{} if and only if there is some $R\in\big[\lfloor\tfrac{|\agents|-1}{2}\rfloor\big]_0$ with
    \[
        \DP_1[R] \;>\; q\cdot|\agents|.
    \]
    Our choice of~$R$ guarantees that blue is the (strict) majority color, so an agent is under majority illusion exactly when it has a majority-red neighborhood; since $K_{1\to\ell}=\agents$, the value $\DP_1[R]$ then equals the number of agents under majority illusion. Correctness thus follows from the next claim for $j=1$.

    \begin{claim}
    \label{clm:MI:ced:DP-correct}
        For all $j \in [\ell]$ and $R \in [\lfloor|\agents|/2\rfloor]_0$, we have $\DP_j[R]=w$ (with $w\in\mathbb{N}$) if and only if there exists a labeling~$f$ of~$G$ extending~$f_W$ such that
        \begin{enumerate}[label=(\alph*),left=2pt]
            \item $\RedNo_f(K_{j \to \ell})=R$, and
            \item $\mrn[f](K_{j \to \ell})=w$.
        \end{enumerate}
    \end{claim}
    \begin{claimproof}
        The proof is by induction on decreasing~$j$ and is identical in structure to that of \Cref{clm:MI:dcg:DP-correct}, with two simplifications: there is no condition on the vector~$\myvec r$, and---by property~(ii)---the contribution $\mrn[f\cup f_W](\vset{K_j})$ of a single clique is fully determined by~$f_W$ and by $\RedNo_f(\vset{K_j})$, the status of every agent of $W\cap\vset{K_j}$ included.

        For $j=\ell$, given a labeling~$f$ satisfying (a)--(b), the definition of $\mathcal{F}_\ell^{R}$ provides a representative whose restriction to the agents in~$K_\ell \setminus W$ has the same number of reds as~$f$; by \Cref{clm:MI:XP:dcg:twinsExchange}, it also yields the same number of agents with a majority-red neighborhood.
        This implies $\DP_\ell[R]\ge w$, and the converse direction is immediate by extending any maximizing representative arbitrarily over~$\agents$.

        For $j\in[\ell-1]$, suppose that $f$ satisfies (a)--(b), and let $R'=\RedNo_f(\vset{K_j})$. The restriction of~$f$ to~$\vset{K_j}$ is equivalent to some $f'\in\mathcal{F}_j^{R'}$ with $\mrn[f'\cup f_W](\vset{K_j})=\mrn[f](\vset{K_j})$, and $\RedNo_f(K_{j+1\to\ell})=R-R'$; by the induction hypothesis $\DP_{j+1}[R-R']\ge \mrn[f](K_{j+1\to\ell})$, so the right-hand side of (\ref{eq:cdn-DP}) is at least $\mrn[f](\vset{K_j})+\mrn[f](K_{j+1\to\ell})=w$. Conversely, if~$R'$ and $f\in\mathcal{F}_j^{R'}$ attain the maximum~$w$ in the recurrence, the induction hypothesis supplies a labeling~$f'$ of~$G$ witnessing $\DP_{j+1}[R-R']$; combining $f'$ on $K_{j+1\to\ell}$, $f$ on the agents of~$K_j \setminus W$, and~$f_W$ on~$W$ gives a labeling of~$G$ satisfying (a)--(b), because every clique's contribution is computed independently of the others' once~$f_W$ is fixed.
    \end{claimproof}

    \proofsubparagraph{Running Time} For a fixed~$f_W$ and clique~$K_j$, the value $\mrn[f\cup f_W](\vset{K_j})$ of each of the $\Oh{|\vset{K_j}|}$ relevant labelings is computed in $\Oh{|\agents|}$ time, so all per-clique contributions are tabulated in $\Oh{|\agents|^2}$ time overall. Each of the $\Oh{|\agents|}$ tables has $\Oh{|\agents|}$ cells, and each cell is evaluated by taking a maximum over $\Oh{|\agents|}$ values of~$R'$; thus the dynamic program runs in $\Oh{|\agents|^3}$ time per choice of~$f_W$ and uses $\Oh{|\agents|}$ space, as only two consecutive tables need to be stored. As there are at most $2^{2k}$ labelings~$f_W$ of~$W$, the total running time is $2^{\Oh{k}}\cdot|\agents|^{\Oh{1}}$.
\end{proof}

\section{Conclusions}

We investigated how the complexity of \qMI changes with various parameters of the underlying social network~$G$; our results outline the limits of tractability in terms of such graph parameters, providing us with efficient algorithms and, in many cases, matching lower bounds. Our results show that to render the problem tractable, one of three conditions must be satisfied: 
a) $G$ is within bounded \emph{edge}-distance from some graph class where the problem is polynomial-time solvable (as in feedback-edge set or cluster edge deletion number), 
b) $G$ is within a bounded \emph{vertex}-distance from components of bounded size (as in vertex integrity), or 
c) $G$ consists of a bounded number of well-structured blocks (as in neighborhood diversity).

We left open the complexity of \qMI when parameterized by the \emph{twin-cover number} of~$G$: by \Cref{thm:MI:XP:dcg}, the problem is in \XP when parameterized by the number of agents whose removal yields a cluster graph, and by \Cref{thm:MI:FPT:cdn}, the problem is in \FPT when parameterized by the number of edges whose removal yields a cluster graph. As the twin-cover number is between these two parameterizations, it would be interesting to see if this additional twin structure of the cliques allows for an \FPT algorithm. More generally, what kind of other structural properties allow for \FPT algorithms for our problem?

\section*{Acknowledgments}

This project was co-funded by the European Union under the project Robotics and Advanced Industrial Production (reg. no. CZ.02.01.01/00/22\_008/0004590). The second author is supported by the Hungarian Academy of Sciences under its Momentum Programme (LP2021-2) and its J\'anos Bolyai Research Scholarship.
Most of this work was done while the first author was at AGH University of Krakow. 

\bibliography{references}

\begin{thebibliography}{42}
\providecommand{\natexlab}[1]{#1}

\bibitem[{Asch(1951)}]{Asch1951}
Asch, S. 1951.
\newblock Effects of Group Pressure Upon the Modification and Distortion of
  Judgments.
\newblock In \emph{Groups, leadership and men; research in human relations},
  177--190.

\bibitem[{Auletta, Ferraioli, and Greco(2020)}]{AulettaFG2020}
Auletta, V.; Ferraioli, D.; and Greco, G. 2020.
\newblock On the Complexity of Reasoning About Opinion Diffusion Under Majority
  Dynamics.
\newblock \emph{Artificial Intelligence}, 284: 103288.

\bibitem[{Banerjee(1992)}]{Banerjee1992}
Banerjee, A.~V. 1992.
\newblock A Simple Model of Herd Behavior.
\newblock \emph{The Quarterly Journal of Economics}, 107(3): 797--817.

\bibitem[{Bodlaender et~al.(2020)Bodlaender, Hanaka, Jaffke, Ono, Otachi, and
  van~der Zanden}]{BodlaenderHJOOZ2020}
Bodlaender, H.~L.; Hanaka, T.; Jaffke, L.; Ono, H.; Otachi, Y.; and van~der
  Zanden, T.~C. 2020.
\newblock Hedonic Seat Arrangement Problems.
\newblock In Seghrouchni, A. E.~F.; Sukthankar, G.; An, B.; and Yorke{-}Smith,
  N., eds., \emph{Proceedings of the 19th International Conference on
  Autonomous Agents and Multiagent Systems, {AAMAS}~'20}, 1777--1779. IFAAMAS.

\bibitem[{Bredereck, Chen, and Woeginger(2013)}]{BredereckCW2013}
Bredereck, R.; Chen, J.; and Woeginger, G.~J. 2013.
\newblock Are There Any Nicely Structured Preference Profiles Nearby?
\newblock In Rossi, F., ed., \emph{Proceedings of the 23rd International Joint
  Conference on Artificial Intelligence, {IJCAI}~'13}, 62--68. {IJCAI/AAAI}.

\bibitem[{Bredereck and Elkind(2017)}]{BredereckE2017}
Bredereck, R.; and Elkind, E. 2017.
\newblock Manipulating Opinion Diffusion in Social Networks.
\newblock In Sierra, C., ed., \emph{Proceedings of the 20th International Joint
  Conference on Artificial Intelligence, {IJCAI}~'17}, 894--900. ijcai.org.

\bibitem[{Bredereck, Jacobs, and Kellerhals(2020)}]{BredereckJK2020}
Bredereck, R.; Jacobs, L.; and Kellerhals, L. 2020.
\newblock Maximizing the Spread of an Opinion in Few Steps: Opinion Diffusion
  in Non-Binary Networks.
\newblock In Bessiere, C., ed., \emph{Proceedings of the 29th International
  Joint Conference on Artificial Intelligence, {IJCAI}~'20}, 1622--1628.
  ijcai.org.

\bibitem[{Castiglioni et~al.(2021)Castiglioni, Ferraioli, Gatti, and
  Landriani}]{CastiglioniFGL2021}
Castiglioni, M.; Ferraioli, D.; Gatti, N.; and Landriani, G. 2021.
\newblock Election Manipulation on Social Networks: Seeding, Edge Removal, Edge
  Addition.
\newblock \emph{Journal of Artificial Intelligence Research}, 71: 1049--1090.

\bibitem[{Cialdini and Goldstein(2004)}]{CialdiniG2004}
Cialdini, R.~B.; and Goldstein, N.~J. 2004.
\newblock Social Influence: {C}ompliance and Conformity.
\newblock \emph{Annual Review of Psychology}, 55: 591--621.

\bibitem[{Cinelli et~al.(2021)Cinelli, Morales, Galeazzi, Quattrociocchi, and
  Starnini}]{CinelliMGQS2021}
Cinelli, M.; Morales, G. D.~F.; Galeazzi, A.; Quattrociocchi, W.; and Starnini,
  M. 2021.
\newblock The Echo Chamber Effect on Social Media.
\newblock \emph{Proceedings of the National Academy of Sciences of the USA},
  118(9): e2023301118.

\bibitem[{Cygan et~al.(2015)Cygan, Fomin, Kowalik, Lokshtanov, Marx, Pilipczuk,
  Pilipczuk, and Saurabh}]{cygan2015parameterized}
Cygan, M.; Fomin, F.~V.; Kowalik, {\L}.; Lokshtanov, D.; Marx, D.; Pilipczuk,
  M.; Pilipczuk, M.; and Saurabh, S. 2015.
\newblock \emph{Parameterized Algorithms}.
\newblock Springer.

\bibitem[{Dippel et~al.(2025)Dippel, la~Tour, Niu, Roy, and
  Vetta}]{DippelTNRV2025}
Dippel, J.; la~Tour, M.~D.; Niu, A.; Roy, S.; and Vetta, A. 2025.
\newblock Eliminating Majority Illusion Is Easy.
\newblock In Walsh, T.; Shah, J.; and Kolter, Z., eds., \emph{Proceedings of
  the 39th AAAI Conference on Artificial Intelligenc {AAAI}~'25}, 13763--13770.
  {AAAI} Press.

\bibitem[{Dom et~al.(2008)Dom, Lokshtanov, Saurabh, and Villanger}]{DomLSV2008}
Dom, M.; Lokshtanov, D.; Saurabh, S.; and Villanger, Y. 2008.
\newblock Capacitated Domination and Covering: {A} Parameterized Perspective.
\newblock In Grohe, M.; and Niedermeier, R., eds., \emph{Proceedings of the 3rd
  International Workshop on Parameterized and Exact Computation, {IWPEC}~'08},
  volume 5018 of \emph{Lecture Notes in Computer Science}, 78--90. Berlin,
  Heidelberg: Springer.

\bibitem[{Doucette et~al.(2019)Doucette, Tsang, Hosseini, Larson, and
  Cohen}]{DoucetteTHLC2019}
Doucette, J.~A.; Tsang, A.; Hosseini, H.; Larson, K.; and Cohen, R. 2019.
\newblock Inferring True Voting Outcomes in Homophilic Social Networks.
\newblock \emph{Autonomous Agents and Multiagent Systems}, 33(3): 298--329.

\bibitem[{Dvoř{\'{a}}k et~al.(2017)Dvoř{\'{a}}k, Eiben, Ganian, Knop, and
  Ordyniak}]{DvorakEGKO2017}
Dvoř{\'{a}}k, P.; Eiben, E.; Ganian, R.; Knop, D.; and Ordyniak, S. 2017.
\newblock Solving Integer Linear Programs With a Small Number of Global
  Variables and Constraints.
\newblock In Sierra, C., ed., \emph{Proceedings of the 26th International Joint
  Conference on Artificial Intelligence, {IJCAI}~'17}, 607--613. ijcai.org.

\bibitem[{Eiben, Ganian, and Ordyniak(2018)}]{EibenGO2018}
Eiben, E.; Ganian, R.; and Ordyniak, S. 2018.
\newblock A Structural Approach to Activity Selection.
\newblock In Lang, J., ed., \emph{Proceedings of the 27th International Joint
  Conference on Artificial Intelligence, {IJCAI}~'18}, 203--209. ijcai.org.

\bibitem[{Eisenbrand et~al.(2025)Eisenbrand, Hunkenschr{\"{o}}der, Klein,
  Kouteck{\'{y}}, Levin, and Onn}]{EisenbrandHKKLO2025}
Eisenbrand, F.; Hunkenschr{\"{o}}der, C.; Klein, K.; Kouteck{\'{y}}, M.; Levin,
  A.; and Onn, S. 2025.
\newblock Sparse Integer Programming Is Fixed-Parameter Tractable.
\newblock \emph{Mathematics of Operations Research}, 50(3): 2141--2156.

\bibitem[{Faliszewski et~al.(2018)Faliszewski, Gonen, Kouteck{\'{y}}, and
  Talmon}]{FaliszewskiGKT2018}
Faliszewski, P.; Gonen, R.; Kouteck{\'{y}}, M.; and Talmon, N. 2018.
\newblock Opinion Diffusion and Campaigning on Society Graphs.
\newblock In Lang, J., ed., \emph{Proceedings of the 27th International Joint
  Conference on Artificial Intelligence, {IJCAI}~'18}, 219--225. ijcai.org.

\bibitem[{Ferrara et~al.(2016)Ferrara, Varol, Davis, Menczer, and
  Flammini}]{FerraraVDMF2016}
Ferrara, E.; Varol, O.; Davis, C.~A.; Menczer, F.; and Flammini, A. 2016.
\newblock The Rise of Social Bots.
\newblock \emph{Communications of the ACM}, 59(7): 96--104.

\bibitem[{Fioravantes, Gahlawat, and Melissinos(2025)}]{FioravantesGM2025}
Fioravantes, F.; Gahlawat, H.; and Melissinos, N. 2025.
\newblock Exact Algorithms and Lower Bounds for Forming Coalitions of
  Constrained Maximum Size.
\newblock In Walsh, T.; Shah, J.; and Kolter, Z., eds., \emph{Proceedings of
  the 39th AAAI Conference on Artificial Intelligence, {AAAI}~'25'},
  13847--13855. {AAAI} Press.

\bibitem[{Fioravantes et~al.(2026)Fioravantes, Gahlawat, Melissinos, and
  Schierreich}]{FioravantesGMS2026}
Fioravantes, F.; Gahlawat, H.; Melissinos, N.; and Schierreich, {\v{S}}. 2026.
\newblock Individual Rationality in Constrained Hedonic Games: Additively
  Separable and Fractional Preferences.
\newblock In \emph{Proceedings of the 25th International Conference on
  Autonomous Agents and Multiagent Systems, {AAMAS}~'26}, 2848--2857. Richland,
  SC: IFAAMAS.

\bibitem[{Fioravantes et~al.(2025{\natexlab{a}})Fioravantes, Knop, Křišťan,
  Melissinos, Opler, and Vu}]{FioravantesKKMO2025a}
Fioravantes, F.; Knop, D.; Křišťan, J.~M.; Melissinos, N.; Opler, M.; and
  Vu, T.~A. 2025{\natexlab{a}}.
\newblock Solving Multiagent Path Finding on Highly Centralized Networks.
\newblock In Walsh, T.; Shah, J.; and Kolter, Z., eds., \emph{Proceedings of
  the 39th AAAI Conference on Artificial Intelligence, {AAAI}~'25'},
  23186--23193. {AAAI} Press.

\bibitem[{Fioravantes et~al.(2025{\natexlab{b}})Fioravantes, Lahiri, Lauerbach,
  Sabater, Sieper, and Wolf}]{FioravantesLLSSW2025}
Fioravantes, F.; Lahiri, A.; Lauerbach, A.; Sabater, L.; Sieper, M.~D.; and
  Wolf, S. 2025{\natexlab{b}}.
\newblock Eliminating Majority Illusion.
\newblock In Das, S.; Now{\'{e}}, A.; and Vorobeychik, Y., eds.,
  \emph{Proceedings of the 24th International Conference on Autonomous Agents
  and Multiagent Systems, {AAMAS}~'25}, 749--757. IFAAMAS.

\bibitem[{Ganian, Klute, and Ordyniak(2021)}]{GanianKO2021}
Ganian, R.; Klute, F.; and Ordyniak, S. 2021.
\newblock On Structural Parameterizations of the Bounded-Degree Vertex Deletion
  Problem.
\newblock \emph{Algorithmica}, 83: 297--336.

\bibitem[{Ganian and Korchemna(2021)}]{GanianK2021}
Ganian, R.; and Korchemna, V. 2021.
\newblock The Complexity of Bayesian Network Learning: Revisiting the
  Superstructure.
\newblock In Ranzato, M.; Beygelzimer, A.; Dauphin, Y.~N.; Liang, P.; and
  Vaughan, J.~W., eds., \emph{Proceedings of the 34th Annual Conference on
  Neural Information Processing Systems 2021, NeurIPS~'21}, 430--442.

\bibitem[{Garimella et~al.(2018)Garimella, Morales, Gionis, and
  Mathioudakis}]{GarimellaMGM2018}
Garimella, K.; Morales, G. D.~F.; Gionis, A.; and Mathioudakis, M. 2018.
\newblock Political Discourse on Social Media: Echo Chambers, Gatekeepers, and
  the Price of Bipartisanship.
\newblock In Champin, P.; Gandon, F.; Lalmas, M.; and Ipeirotis, P.~G., eds.,
  \emph{Proceedings of the 2018 World Wide Web Conference, {WWW}~'18},
  913--922. {ACM}.

\bibitem[{Grandi(2017)}]{Grandi2017}
Grandi, U. 2017.
\newblock Social Choice and Social Networks.
\newblock In Endriss, U., ed., \emph{Trends in Computational Social Choice},
  169--184. AI Access.

\bibitem[{Grandi et~al.(2025)Grandi, Kanesh, Lisowski, Ramanujan, and
  Turrini}]{GrandiKLRT2025}
Grandi, U.; Kanesh, L.; Lisowski, G.; Ramanujan, M.~S.; and Turrini, P. 2025.
\newblock A Complexity-Theoretic Analysis of Majority Illusion in Social
  Networks.
\newblock \emph{Journal of Artificial Intelligence Research}, 83: 26.

\bibitem[{Gr{\"{u}}ttemeier and Komusiewicz(2020)}]{GruttemeierK2020}
Gr{\"{u}}ttemeier, N.; and Komusiewicz, C. 2020.
\newblock Learning Bayesian Networks Under Sparsity Constraints: {A}
  Parameterized Complexity Analysis.
\newblock In Bessiere, C., ed., \emph{Proceedings of the 29th International
  Joint Conference on Artificial Intelligence, {IJCAI}~'20}, 4245--4251.
  ijcai.org.

\bibitem[{Guo et~al.(2022)Guo, Li, Neumann, Neumann, and Nguyen}]{GuoLN0N2022}
Guo, M.; Li, J.; Neumann, A.; Neumann, F.; and Nguyen, H. 2022.
\newblock Practical Fixed-Parameter Algorithms for Defending Active Directory
  Style Attack Graphs.
\newblock In \emph{Proceedings of the 36th {AAAI} Conference on Artificial
  Intelligence, {AAAI}~'22}, 9360--9367. {AAAI} Press.

\bibitem[{Jana and Roy(2026)}]{JanaR2026}
Jana, S.; and Roy, S. 2026.
\newblock Eliminating Illusion in Directed Networks.
\newblock arXiv:2604.02395.

\bibitem[{Kempe, Kleinberg, and Tardos(2015)}]{KempeKT2015}
Kempe, D.; Kleinberg, J.~M.; and Tardos, {\'{E}}. 2015.
\newblock Maximizing the Spread of Influence Through a Social Network.
\newblock \emph{Theory of Computing}, 11: 105--147.

\bibitem[{Knop, Schierreich, and Suchý(2026)}]{KnopSS2026}
Knop, D.; Schierreich, {\v{S}}.; and Suchý, O. 2026.
\newblock Balancing the Spread of Two Opinions in Sparse Social Networks.
\newblock \emph{Artificial Intelligence}, 357: 104563.

\bibitem[{Lerman, Yan, and Wu(2016)}]{LermanYW2016}
Lerman, K.; Yan, X.; and Wu, X.-Z. 2016.
\newblock The "Majority Illusion" in Social Networks.
\newblock \emph{PLoS ONE}, 11(2): 1--13.

\bibitem[{Musco, Musco, and Tsourakakis(2018)}]{MuscoMT2018}
Musco, C.; Musco, C.; and Tsourakakis, C.~E. 2018.
\newblock Minimizing Polarization and Disagreement in Social Networks.
\newblock In Champin, P.; Gandon, F.; Lalmas, M.; and Ipeirotis, P.~G., eds.,
  \emph{Proceedings of the 2018 World Wide Web Conference, {WWW}~'18},
  369--378. {ACM}.

\bibitem[{Schierreich(2023)}]{Schierreich2023}
Schierreich, {\v{S}}. 2023.
\newblock Maximizing Influence Spread Through a Dynamic Social Network (Student
  Abstract).
\newblock In Williams, B.; Chen, Y.; and Neville, J., eds., \emph{Proceedings
  of the 37th {AAAI} Conference on Artificial Intelligence, {AAAI}~'23},
  16316--16317. {AAAI} Press.

\bibitem[{UNESCO(2023)}]{UNESCO2023}
UNESCO. 2023.
\newblock Survey on the Impact of Online Disinformation and Hate Speech.
\newblock Technical report, United Nations Educational, Scientific and Cultural
  Organization.

\bibitem[{Venema{-}Los, Christoff, and Grossi(2025)}]{VenemaLosCG2025}
Venema{-}Los, M.; Christoff, Z.; and Grossi, D. 2025.
\newblock On the Graph Theory of Majority Illusions: Theoretical Results and
  Computational Experiments.
\newblock \emph{Autonomous Agents and Multi-Agent Systems}, 39(2): 39.

\bibitem[{Vosoughi, Roy, and Aral(2018)}]{VosoughiRA2018}
Vosoughi, S.; Roy, D.; and Aral, S. 2018.
\newblock The Spread of True and False News Online.
\newblock \emph{Science}, 359(6380): 1146--1151.

\bibitem[{Wilder and Vorobeychik(2018)}]{WilderV2018}
Wilder, B.; and Vorobeychik, Y. 2018.
\newblock Controlling Elections Through Social Influence.
\newblock In Andr{\'{e}}, E.; Koenig, S.; Dastani, M.; and Sukthankar, G.,
  eds., \emph{Proceedings of the 17th International Conference on Autonomous
  Agents and MultiAgent Systems, {AAMAS}~'18}, 265--273. IFAAMAS.

\bibitem[{Woolley and Howard(2018)}]{WoolleyH2018}
Woolley, S.~C.; and Howard, P.~N. 2018.
\newblock \emph{Computational Propaganda: Political Parties, Politicians, and
  Political Manipulation on Social Media}.
\newblock Oxford University Press.
\newblock ISBN 9780190931407.

\bibitem[{Yang et~al.(2020)Yang, Varol, Hui, and Menczer}]{YangVHM2020}
Yang, K.; Varol, O.; Hui, P.; and Menczer, F. 2020.
\newblock Scalable and Generalizable Social Bot Detection Through Data
  Selection.
\newblock In \emph{Proceedings of the 34th {AAAI} Conference on Artificial
  Intelligence, {AAAI}~'20}, 1096--1103. {AAAI} Press.

\end{thebibliography}

\end{document}